\documentclass[11pt]{article}
\usepackage[margin=1.5in]{geometry}
\usepackage[utf8]{inputenc}
\usepackage{amsmath,amsfonts}
\usepackage{algorithm, algpseudocode}
\usepackage{bbm,color,verbatim}
\usepackage{amsthm}
\newcommand{\mutual}{\ensuremath{{I}}}
\newcommand{\entropy}{\ensuremath{{H}}}
\newtheorem{lemma}{Lemma}
\newtheorem{proposition}{Proposition}

\newtheorem{theorem}{Theorem}
\newtheorem{corollary}{Corollary}
\theoremstyle{definition}

\newcommand{\citeN}{\citet}
\newcommand*{\defeq}{\equiv}

\newcommand{\expt}[1]{{\mathbb{E}}[#1]}
\newcommand{\norm}[1]{\lVert#1\rVert}

\newcommand{\ip}[2]{\langle #1, #2 \rangle}

\renewcommand{\epsilon}{\varepsilon}

\DeclareMathOperator*{\argmax}{arg\,max}
\DeclareMathOperator*{\argmin}{arg\,min}
\newcommand{\reals}{\ensuremath{\mathbb{R}}}

\newcommand{\prob}{\ensuremath{\mathrm{Pr}}}
\newcommand{\stratis}[1]{\textcolor{green}{Stratis: #1}}
\newcommand{\thibaut}[1]{\textcolor{blue}{Thibaut: #1}}
\newcommand{\T}[1]{#1^T}
\newcommand{\EDP}{EDP}
\newcommand{\SEDP}{EDP}
\newcommand{\E}{{\tt E}}
\newcommand{\id}{\mathbbm{1}}
\newcommand{\junk}[1]{}

\newcommand{\dom}{\mathcal{D}}

\algrenewcommand\algorithmicrequire{\textbf{Input:}}
\algrenewcommand\algorithmicensure{\textbf{Output:}}

\usepackage[pagebackref=true,breaklinks=true,colorlinks=true]{hyperref}
\usepackage[numbers]{natbib}
\title{Budget Feasible Mechanisms for Experimental Design}
\author{
    Thibaut Horel\\École Normale Supérieure\\\texttt{thibaut.horel@normalesup.org}
    \and
    Stratis Ioannidis\\Technicolor\\\texttt{stratis.ioannidis@technicolor.com}
    \and
    S. Muthukrishnan\\Rutgers University--Microsoft Research\\\texttt{muthu@cs.rutgers.edu}
}
\date{}
\begin{document}
\maketitle
\thispagestyle{empty}
\begin{abstract}

In the classical {\em experimental design} setting,
an  experimenter  \E\ 
has access to a population of $n$ potential experiment subjects $i\in \{1,\ldots,n\}$,  each associated with a vector of features $x_i\in\reals^d$.
Conducting an experiment with subject $i$  reveals an unknown value $y_i\in \reals$ to \E. \E\ typically assumes some 
hypothetical relationship between  $x_i$'s and $y_i$'s, \emph{e.g.},  $y_i \approx  \T{\beta} x_i$, and estimates 
$\beta$ from experiments, \emph{e.g.}, through linear regression. 
As a proxy for various practical constraints, \E{} may select only a subset of subjects on which to conduct the experiment. 

We initiate the study of budgeted mechanisms for experimental design. In this setting,  \E{} has a budget $B$.  
Each subject $i$ declares an associated cost $c_i >0$ to be part of the experiment, and must be paid at least her cost.  In particular,  the {\em  Experimental Design Problem} (\SEDP) is to find a set $S$ of subjects for the experiment that maximizes $V(S) = \log\det(I_d+\sum_{i\in S}x_i\T{x_i})$ under the constraint $\sum_{i\in S}c_i\leq B$; our objective function corresponds to  the information gain in  parameter  $\beta$ that is learned through linear regression methods, and is related to the so-called $D$-optimality criterion.  Further,  the subjects  are \emph{strategic} and may lie about their costs. Thus, we need to design a 
mechanism for \SEDP{} with suitable properties. 

We present a deterministic, polynomial time, budget feasible mechanism scheme, that is approximately truthful and yields a constant ($\approx 12.98$) factor  approximation to \EDP. 
By applying previous work on budget feasible mechanisms with a submodular objective, one could {\em only} have derived either an exponential time deterministic mechanism or a randomized  polynomial time mechanism. We also establish that no truthful, budget-feasible mechanism is possible within a factor $2$ approximation, and show how to generalize our approach to a wide class of learning problems, beyond linear regression.

\end{abstract}

\clearpage
\setcounter{page}{1}
\section{Introduction}
 In the classic setting of experimental design \cite{pukelsheim2006optimal,atkinson2007optimum},
 an {\em experimenter}  \E\ has access to a population of $n$ potential experiment subjects. 
Each subject $i\in  \{1,\ldots,n\}$ is associated with a set of parameters (or features) $x_i\in \reals^d$, 
known to the experimenter. 
\E\ wishes to measure a certain inherent property of the subjects by performing an experiment: the outcome $y_i$ of the experiment on a subject $i$ is unknown to \E\ before the experiment is performed.

Typically, \E\ has a hypothesis on the relationship between $x_i$'s and $y_i$'s. Due to its simplicity, as well as its ubiquity in statistical analysis, a large body of work has focused on linear hypotheses: \emph{i.e.}, it is assumed that there exists a $\beta\in\reals^d$ such that  
$$y_i =  \T{\beta} x_i+\varepsilon_i,$$ for all $i\in \{1,\ldots,n\},$ where $\varepsilon_i$ are zero-mean, i.i.d.~random variables. Conducting the experiments and obtaining the measurements $y_i$ lets \E\  estimate  $\beta$, \emph{e.g.}, through linear regression. 
 
The above experimental design scenario  has many applications. Regression over  personal data collected through surveys or experimentation is the cornerstone of marketing research, as well as research in a variety of experimental sciences such as medicine and sociology. Crucially, statistical analysis of user data is also a widely spread practice among Internet companies, which routinely use machine learning techniques over vast records of user data to perform inference and classification tasks integral to their daily operations.
Beyond linear regression, there is a rich literature about estimation procedures, as well as about means of quantifying the quality of the produced estimate~\cite{pukelsheim2006optimal}.  There is also an extensive theory on how to select subjects 
if \E\ can conduct only a limited number of experiments, so the estimation process returns a $\beta$
that approximates the true parameter of the underlying population \cite{ginebra2007measure,le1996comparison,chaloner1995bayesian,boyd2004convex}. 

We depart from this classical setup by viewing experimental design in a strategic setting, and by studying budgeted mechanism design issues. 
In our setting, experiments cannot be manipulated and hence measurements are reliable. 
 \E{} has a total budget of $B$ to conduct all the experiments. 
There is a cost $c_i$ associated with experimenting on
subject $i$ which varies from subject to subject.  This cost $c_i$ is determined by the subject $i$ and reported to \E; subjects are strategic and may misreport these costs. Intuitively, $c_i$  may be viewed as the  
cost $i$ incurs when tested and for which she needs to be reimbursed; or, it might be viewed as the incentive for $i$ to participate in the experiment; or, it might be the intrinsic worth of the data to the subject. 
 The economic aspect of paying subjects has always been inherent in experimental design: experimenters often work within strict budgets and design creative incentives. Subjects often negotiate better incentives or higher payments. 
However, we are not aware of a principled study of this setting from a strategic point of view, when subjects declare their costs and therefore determine their payment.  Such a setting is increasingly realistic, given the growth of these experiments over the Internet. 


Our contributions are as follows.
\begin{itemize}
\item
We initiate the study of experimental design in the presence of a budget and strategic subjects. 
 In particular, we formulate the  {\em  Experimental Design Problem} (\SEDP) as
 follows: the experimenter \E\ wishes to find a set $S$ of subjects to maximize 
\begin{align}V(S) = \log\det\Big(I_d+\sum_{i\in S}x_i\T{x_i}\Big) \label{obj}\end{align}
subject to a budget constraint $\sum_{i\in S}c_i\leq B$, where $B$ is \E's budget. When subjects are strategic, the above problem can be naturally approached  as a \emph{budget feasible mechanism design} problem, as introduced by \citeN{singer-mechanisms}.

\smallskip
The objective function, which is the key, is formally obtained by optimizing  the information gain in  $\beta$ when the latter is learned  through ridge regression, and is related to  the so-called $D$-optimality criterion~\cite{pukelsheim2006optimal,atkinson2007optimum}. 
\item
We present a polynomial time mechanism scheme for \SEDP{} that is approximately truthful and yields a constant factor ($\approx 12.98$) approximation to the optimal value of \eqref{obj}. 
In contrast to this, we show that no truthful, budget-feasible mechanisms are possible for \SEDP{}  within a factor 2 approximation. 

\smallskip
We note that the objective \eqref{obj} is submodular. Using this fact, applying
previous results on budget feasible mechanism design under general submodular
objectives~\cite{singer-mechanisms,chen} would yield either a deterministic,
truthful, constant-approximation mechanism that requires exponential time,  or
a non-determi\-nistic, (universally) truthful, poly-time mechanism that yields a constant approximation ratio only \emph{in expectation} (\emph{i.e.}, its approximation guarantee for a given instance may in fact be unbounded).
\end{itemize}

From a technical perspective, we propose a convex optimization problem and establish that its optimal value is within a constant factor from the optimal value of \EDP.  
 In particular, we show our relaxed objective is within a constant factor from the so-called multi-linear extension of  \eqref{obj}, which in turn can be related to \eqref{obj} through pipage rounding. We establish the constant factor to the multi-linear extension by bounding  the partial derivatives of these two functions; we achieve the latter by exploiting convexity properties of matrix functions over the convex cone of positive semidefinite matrices.

Our convex relaxation of \EDP{} involves maximizing a self-concordant function subject to linear constraints. Its optimal value can be computed with arbitrary accuracy in polynomial time using the so-called barrier method. However, the outcome of this computation may not necessarily be monotone, a property needed in designing a truthful mechanism. Nevetheless, we construct an algorithm that solves the above convex relaxation and is ``almost'' monotone; we achieve this by applying the barrier method on a set perturbed constraints, over which our objective is ``sufficiently'' concave. In turn, we show how to employ this algorithm to design a poly-time, $\delta$-truthful, constant-approximation mechanism for \EDP{}.



In what follows, we describe related work in Section~\ref{sec:related}. We briefly review  experimental design and budget feasible mechanisms in Section~\ref{sec:peel} and define \SEDP\ formally. We present our convex relaxation to \EDP{} in Section~\ref{sec:approximation} and use it to construct our mechanism in Section~\ref{sec:main}. We conclude in Section~\ref{sec:concl}. All proofs of our technical results are provided in the appendix. 

\junk{

\stratis{maximizing other ``optimality criteria'' (E-optimality, G-optimality, etc.) Improved upper lower bound for $R\neq I$. General entropy gain optimization}

\begin{itemize}
    \item already existing field of experiment design: survey-like setup, what
    are the best points to include in your experiment? Measure of the
    usefulness of the data: variance-reduction or entropy-reduction.
    \item nowadays, there is also a big focus on purchasing data: paid surveys,
    mechanical turk, etc. that add economic aspects to the problem of
    experiment design
    \item recent advances (Singer, Chen) in the field of budgeted mechanisms
    \item we study ridge regression, very widely used in statistical learning,
    and treat it as a problem of budgeted experiment design
    \item we make the following contributions: ...
    \item extension to a more general setup which includes a wider class of
    machine learning problems
\end{itemize}

}

\section{Related work}
\label{sec:related}

\junk{\subsection{Experimental Design} The classic experimental design problem, which we also briefly review in Section~\ref{sec:edprelim}, deals with  which $k$ experiments to conduct among a set of $n$ possible experiments. It is a well studied problem both in the non-Bayesian \cite{pukelsheim2006optimal,atkinson2007optimum,boyd2004convex} and Bayesian setting \cite{chaloner1995bayesian}. Beyond $D$-optimality, several other objectives are encountered in the literature  \cite{pukelsheim2006optimal}; many involve some function of the covariance matrix of the estimate of $\beta$, such as $E$-optimality (maximizing the smallest eigenvalue of the covariance of $\beta$) or $T$-optimality (maximizing the trace). Our focus on $D$-optimality is motivated by both its tractability as well as its relationship to the information gain.  
}
\paragraph{Budget Feasible Mechanisms for General Submodular Functions}
Budget feasible mechanism design was originally proposed by  \citeN{singer-mechanisms}. Singer considers the problem of maximizing an arbitrary submodular function subject to a budget constraint in the \emph{value query} model, \emph{i.e.} assuming an  oracle providing the  value of the submodular objective on any given set. 
 Singer shows that there exists a randomized, 112-approximation mechanism for submodular maximization that is \emph{universally truthful} (\emph{i.e.}, it is a randomized mechanism sampled from a distribution over truthful mechanisms). \citeN{chen}  improve this result by providing a 7.91-approximate mechanism, and show a corresponding lower bound of $2$ among universally truthful randomized mechanisms for submodular maximization.

The above approximation guarantees hold for the expected value of the
randomized mechanism: for a given
instance, the approximation ratio provided by the mechanism may in fact be
unbounded. No deterministic, truthful, constant approximation mechanism that
runs in polynomial time is presently known for submodular maximization.
However, assuming access to an oracle providing the optimum in the
full-information setup, Chen \emph{et al.},~propose a truthful, $8.34$-approximate mechanism; in cases for which the full information problem is NP-hard, as the one we consider here, this mechanism is not poly-time, unless P=NP.  Chen \emph{et al.}~also prove a $1+\sqrt{2}$ lower bound for truthful deterministic mechanisms, improving upon an earlier bound of 2 by \citeN{singer-mechanisms}.

\paragraph{Budget Feasible Mechanisms on Specific Problems}
Improved bounds, as well as deterministic polynomial mechanisms, are known for specific submodular objectives. For symmetric submodular functions, a truthful mechanism with approximation ratio 2 is known, and this ratio is tight \cite{singer-mechanisms}. Singer also provides a 7.32-approximate truthful mechanism for the budget feasible version of \textsc{Matching}, and a corresponding lower bound of 2 \cite{singer-mechanisms}. Improving an earlier result by Singer, \citeN{chen} give a truthful, $2+\sqrt{2}$-approximate mechanism for \textsc{Knapsack}, and a lower bound of $1+\sqrt{2}$. Finally, a truthful, 31-approximate  mechanism is also known for the budgeted version of \textsc{Coverage} \cite{singer-influence}. 

The  deterministic mechanisms for \textsc{Knapsack} \cite{chen} and
\textsc{Coverage} \cite{singer-influence} follow the same general framework,
which we also employ in our mechanism for \EDP. We describe this framework in detail in
Section~\ref{sec:main}. Both of these mechanisms rely on approximating the
optimal solution to the underlying combinatorial problem by a well-known linear program (LP) relaxation~\cite{pipage}, which can be solved exactly in polynomial time. No
such  relaxation exists for \EDP, which unlikely to be
approximable through an LP due to its logarithmic objective. We develop instead a convex relaxation to \EDP;
though, contrary to the above LP relaxations, this cannot be solved exactly, we
establish that it can be incorporated in the framework of
\cite{chen,singer-influence} to yield a $\delta$-truthful mechanism for \EDP.


\paragraph{Beyond Submodular Objectives}
Beyond submodular objectives, it is known that no truthful mechanism with approximation ratio smaller than $n^{1/2-\epsilon}$ exists for maximizing  fractionally subadditive functions (a class that includes submodular functions) assuming access to a value query oracle~\cite{singer-mechanisms}. Assuming access to a stronger oracle (the \emph{demand} oracle), there exists
a truthful, $O(\log^3 n)$-approximate mechanism 
\cite{dobz2011-mechanisms} as well as a universally truthful, $O(\frac{\log n}{\log \log n})$-approximate mechanism for subadditive maximization 
\cite{bei2012budget}. Moreover, in a Bayesian setup, assuming a prior distribution among the agent's costs, there exists a truthful mechanism with a 768/512-approximation ratio \cite{bei2012budget}. 
Posted price, rather than direct revelation mechanisms, are also studied in \cite{singerposted}.

\paragraph{Monotone Approximations in Combinatorial Auctions}
Relaxations of combinatorial problems are prevalent
in \emph{combinatorial auctions}, 
in which an auctioneer aims at maximizing a set function which is the sum of utilities of strategic bidders (\emph{i.e.}, the social welfare). As noted by \citeN{archer-approximate},
approximations to this maximization must preserve incentive compatibility and truthfulness. Most approximation
algorithms do not preserve these properties, hence specific relaxations, and corresponding roundings to an integral solution, must be
constructed. \citeN{archer-approximate} propose a randomized rounding of the LP relaxation of the \textsc{SetPacking} problem, yielding a mechanism
which is \emph{truthful-in-expectation}. 
\citeN{lavi-truthful} construct randomized
truthful-in-expectation mechanisms for several combinatorial auctions, improving the approximation
ratio in \cite{archer-approximate}, by treating the fractional solution of an LP as a probability distribution from which they sample integral solutions. 

Beyond LP relaxations, 
\citeN{dughmi2011convex} propose
truthful-in-expectation mechanisms for combinatorial auctions in which
the bidders' utilities are matroid rank sum functions (applied earlier to the \textsc{CombinatorialPublicProjects} problem \cite{dughmi-truthful}).  Their framework relies
on solving a convex optimization problem which can only be solved approximately. As in \cite{lavi-truthful}, they also treat the fractional solution as a distribution over which they sample integral solutions. The authors ensure that a solver is applied to a
``well-conditioned'' problem,  which resembles the technical challenge we face in
Section~\ref{sec:monotonicity}. However, we seek a deterministic mechanism and $\delta$-truthfulness, not truthfulness-in-expectation. In addition, our objective is not a matroid rank sum function. As such, both the methodology for dealing with problems that are not ``well-conditioned''  as well as the approximation guarantees of the convex relaxation  in \cite{dughmi2011convex} do not readily extend to \EDP.

\citeN{briest-approximation} construct monotone FPTAS for problems that can be approximated through rounding techniques, which in turn can be used to construct truthful, deterministic, constant-approximation  mechanisms for corresponding  combinatorial auctions. \EDP{} is not readily approximable through such rounding techniques; as such, we rely on a relaxation to approximate it.

\paragraph{$\delta$-Truthfulness and Differential Privacy}

The notion of $\delta$-truthfulness has attracted considerable attention recently in the context of differential privacy (see, \emph{e.g.}, the survey by \citeN{pai2013privacy}).  \citeN{mcsherrytalwar} were the first to observe that any $\epsilon$-differentially private mechanism must also be $\delta$-truthful in expectation, for $\delta=2\epsilon$. This property was used to construct $\delta$-truthful (in expectation) mechanisms for a digital goods auction~\cite{mcsherrytalwar} and for $\alpha$-approximate equilibrium selection \cite{kearns2012}. \citeN{approximatemechanismdesign} propose a framework for converting a differentially private mechanism to a truthful-in-expectation mechanism by randomly selecting between a differentially private mechanism with good approximation guarantees, and a truthful mechanism. They apply their framework to the \textsc{FacilityLocation} problem. We depart from the above works in seeking a deterministic mechanism for \EDP, and using a stronger notion of $\delta$-truthfulness.

\fussy

\section{Preliminaries}\label{sec:peel}
\label{sec:prel}
\subsection{Linear Regression and Experimental Design}\label{sec:edprelim}

The theory of experimental design \cite{pukelsheim2006optimal,atkinson2007optimum,chaloner1995bayesian} considers the following formal setting. 
Suppose that an experimenter \E\ wishes to conduct $k$ among $n$ possible
experiments. Each experiment $i\in\mathcal{N}\defeq \{1,\ldots,n\}$ is
associated with a set of parameters (or features) $x_i\in \reals^d$, normalized
so that $$b\leq \|x_i\|^2_2\leq 1,$$ for some $b>0$. Denote by $S\subseteq \mathcal{N}$, where $|S|=k$, the set of experiments selected; upon its execution, experiment $i\in S$ reveals an output variable (the ``measurement'') $y_i$,  related to the experiment features $x_i$ through a linear function, \emph{i.e.},
\begin{align}\label{model}
     \forall i\in\mathcal{N},\quad y_i = \T{\beta} x_i + \varepsilon_i
\end{align}
where $\beta$ is a vector in $\reals^d$, commonly referred to as the \emph{model}, and $\varepsilon_i$ (the \emph{measurement noise}) are independent, normally distributed random  variables with mean 0 and variance $\sigma^2$. 

For example, each $i$ may correspond to a human subject; the
feature vector $x_i$ may correspond to a normalized vector of her age, weight,
gender, income, \emph{etc.}, and the measurement $y_i$ may capture some
biometric information (\emph{e.g.}, her red cell blood count, a genetic marker,
etc.). The magnitude of the coefficient $\beta_i$ captures the effect that feature $i$ has on the measured variable, and its sign captures whether the correlation is positive or negative.

The purpose of these experiments is to allow \E\  to estimate the model $\beta$. In particular,
 assume that the experimenter  \E\ has a {\em prior}
distribution on $\beta$, \emph{i.e.},  $\beta$ has a multivariate normal prior
with zero mean  and covariance $\sigma^2R^{-1}\in \reals^{d^2}$ (where $\sigma^2$ is the noise variance). 
Then, \E\ estimates $\beta$ through \emph{maximum a posteriori estimation}: \emph{i.e.}, finding the parameter which maximizes the posterior distribution of $\beta$ given the observations $y_S$. Under the linearity assumption \eqref{model} and the Gaussian prior on $\beta$, maximum a posteriori estimation leads to the following maximization \cite{hastie}: 
\begin{equation}
\begin{split}
    \hat{\beta} &= \argmax_{\beta\in\reals^d} \prob(\beta\mid y_S) =\argmin_{\beta\in\reals^d} \big(\sum_{i\in S} (y_i - \T{\beta}x_i)^2
    + \T{\beta}R\beta\big)\\
    &= (R+\T{X_S}X_S)^{-1}X_S^Ty_S \label{ridge}
    \end{split}
\end{equation}
where the last equality is obtained by setting  $\nabla_{\beta}\prob(\beta\mid y_S)$ to zero and solving the resulting linear system; in \eqref{ridge}, $X_S\defeq[x_i]_{i\in S}\in \reals^{|S|\times d}$ is the matrix of experiment features and
$y_S\defeq[y_i]_{i\in S}\in\reals^{|S|}$ are the observed measurements. 
This optimization, commonly known as \emph{ridge regression}, includes an additional quadratic penalty term $\beta^TR\beta$ compared to the standard least squares estimation.

Let $V:2^\mathcal{N}\to\reals$ be a \emph{value function}, quantifying how informative a set of experiments $S$ is in estimating $\beta$. The classical experimental design problem amounts to finding a set $S$ that maximizes $V(S)$ subject to the constraint $|S|\leq k$. 
A variety of different value functions are used in literature~\cite{pukelsheim2006optimal,boyd2004convex}; 
one that has natural advantages is the \emph{information gain}:  
\begin{align}
   V(S)= I(\beta;y_S) = \entropy(\beta)-\entropy(\beta\mid y_S). \label{informationgain}
\end{align}
which is the entropy reduction on $\beta$ after the revelation of $y_S$ (also known as the mutual information between $y_S$ and $\beta$). 
Hence, selecting a set of experiments $S$ that
maximizes $V(S)$ is equivalent to finding the set of experiments that minimizes
the uncertainty on $\beta$, as captured by the entropy reduction of its estimator.
Under the linear model \eqref{model}, and the Gaussian prior, the information gain takes the following form (see, \emph{e.g.}, \cite{chaloner1995bayesian}):
\begin{align}
 I(\beta;y_S)&= \frac{1}{2}\log\det(R+ \T{X_S}X_S) - \frac{1}{2}\log\det R\label{dcrit} 
\end{align}
Maximizing $I(\beta;y_S)$ is therefore equivalent to maximizing $\log\det(R+ \T{X_S}X_S)$, which is  known in the experimental design literature as the Bayes
$D$-optimality criterion
\cite{pukelsheim2006optimal,atkinson2007optimum,chaloner1995bayesian}. 



Our analysis will focus on the case of a \emph{homotropic} prior, in which the
prior covariance is the identity matrix, \emph{i.e.}, $R=I_d\in \reals^{d\times
d}.$ Intuitively, this corresponds to the simplest prior, in which no direction
of $\reals^d$ is a priori favored; equivalently, it also corresponds to the
case where ridge regression estimation \eqref{ridge} performed by $\E$ has
a penalty term $\norm{\beta}_2^2$. A generalization of our results to arbitrary 
covariance matrices $R$ can be found in Appendix~\ref{sec:ext}.


\subsection{Budget-Feasible Experimental Design: Full Information Case}\label{sec:fullinfo}

Instead of the cardinality constraint in classical experimental design discussed above, we consider a budget-constrained version. 
Each experiment is associated with a cost $c_i\in\reals_+$.   
The cost $c_i$ can capture, \emph{e.g.}, the amount the subject $i$ deems sufficient to incentivize her participation in the experiment. The experimenter $\E$ is limited by a budget $B\in \reals_+$.  In the full-information case,  experiment costs are common knowledge; as such, the experimenter wishes to solve: 
\medskip\\\hspace*{\stretch{1}}\textsc{ExperimentalDesignProblem} (\EDP)\hspace*{\stretch{1}}
\begin{subequations}
\begin{align}
\text{Maximize}\quad V(S) &= \log\det(I_d+\T{X_S}X_S) \label{modified} \\
\text{subject to}\quad \sum_{i\in S} c_i&\leq B\label{lincon}
\end{align}\label{edp}
\end{subequations}
\emph{W.l.o.g.}, we assume that $c_i\in [0,B]$ for all $i\in \mathcal{N}$, as no $i$ with $c_i>B$ can be in an $S$ satisfying \eqref{lincon}. Denote by
\begin{equation}\label{eq:non-strategic}
    OPT = \max_{S\subseteq\mathcal{N}} \Big\{V(S) \;\Big| \;
    \sum_{i\in S}c_i\leq B\Big\}
\end{equation}
the optimal value achievable in the full-information case. 
\EDP{}, as defined above, is  NP-hard; to see this, note that \textsc{Knapsack}
reduces to EDP with dimension $d=1$ by mapping the weight of each item, say,
$w_i$, to an experiment with $x_i^2=w_i$.

The value function \eqref{modified} has the following properties, which are proved in Appendix~\ref{app:properties}. First, it is non-negative, \emph{i.e.}, $V(S)\geq 0$
for all $S\subseteq\mathcal{N}$. Second, it is
also monotone, \emph{i.e.}, $V(S)\leq V(T)$ for all $S\subseteq T$, with
$V(\emptyset)=0$. Finally, it is submodular, \emph{i.e.}, 
$V(S\cup \{i\})-V(S)\geq V(T\cup \{i\})-V(T)$ for all $S\subseteq T\subseteq \mathcal{N}$ and $i\in \mathcal{N}$. 
The above imply that a greedy algorithm yields a constant approximation ratio to \EDP.
In particular, consider the greedy algorithm in which, for
 $S\subseteq\mathcal{N}$ the set constructed thus far, the next
element $i$  included is the one which maximizes the
\emph{marginal-value-per-cost}, \emph{i.e.},
  $  i = \argmax_{j\in\mathcal{N}\setminus S}{(V(S\cup\{i\}) - V(S))}/{c_i}.$ 
This is repeated until adding an element in $S$ exceeds the budget
 $B$. Denote by $S_G$ the set constructed by this heuristic and let
$i^*=\argmax_{i\in\mathcal{N}} V(\{i\})$ be the element of maximum singleton value. Then,
the following algorithm:
\begin{equation}\label{eq:max-algorithm}
\textbf{if}\; V(\{i^*\}) \geq V(S_G)\; \textbf{return}\; \{i^*\}
\;\textbf{else return}\; S_G
\end{equation}
yields an approximation ratio  of $\frac{5 e }{e-1}~$\cite{singer-mechanisms}; this can be further improved to $\frac{e}{e-1}$ using more complicated greedy set constructions~\cite{krause-submodular,sviridenko-submodular}.

 \subsection{Budget-Feasible Experimental Design: Strategic Case}
We study the following \emph{strategic} setting, in which the costs $c_i$ are {\em not} common knowledge and their reporting can be manipulated by the experiment subjects. The latter are strategic and wish to maximize their utility, which is the difference of the payment they receive and their true cost. We note that, though subjects may misreport $c_i$, they cannot lie about $x_i$ (\emph{i.e.}, all public features are verifiable prior to the experiment) nor $y_i$ (\emph{i.e.}, the subject cannot falsify her measurement).
In this setting,  experimental design reduces to a \emph{budget feasible reverse auction}, as introduced by \citeN{singer-mechanisms}; we review the formal definition  in Appendix~\ref{app:budgetfeasible}. In short, given a budget $B$ and a value function $V:2^{\mathcal{N}}\to\reals_+$, a \emph{reverse auction mechanism} $\mathcal{M} = (S,p)$ comprises (a) an \emph{allocation function}\footnote{Note that $S$ would be more aptly termed as a \emph{selection} function, as this is a reverse auction, but we retain the term ``allocation'' to align with the familiar term from standard auctions.}
$S:\reals_+^n \to 2^\mathcal{N}$,  determining the set 
 of experiments to be purchased, and (b) a \emph{payment function}
$p:\reals_+^n\to \reals_+^n$, determining the payments $[p_i(c)]_{i\in \mathcal{N}}$ received by experiment subjects.
  
We seek mechanisms that are \emph{normalized} (unallocated experiments receive zero payments), \emph{individually rational} (payments for allocated experiments exceed costs), have \emph{no positive transfers} (payments are non-negative), and are \emph{budget feasible} (the sum of payments does not exceed the budget $B$). 
We relax the notion of truthfulness to \emph{$\delta$-truthfulness}, requiring that reporting one's true cost is an \emph{almost-dominant} strategy: no subject increases their utility by reporting a cost that differs more than $\delta>0$ from their true cost. Under this definition, a mechanism is truthful if $\delta=0$. 
In addition, we would like the allocation $S(c)$ to be of maximal value; however, $\delta$-truthfulness, as well as the hardness of \EDP{}, preclude achieving this goal. Hence, we seek mechanisms with that  are \emph{$(\alpha,\beta)$-approximate}, \emph{i.e.}, there exist $\alpha\geq 1$ and $\beta>0$ s.t.~$OPT \leq \alpha V(S(c))+\beta$, and are \emph{computationally efficient}, in that $S$ and $p$ can be computed in polynomial time.

We note that  the constant approximation algorithm \eqref{eq:max-algorithm} breaks
truthfulness. Though this is not true for all submodular functions (see, \emph{e.g.}, \cite{singer-mechanisms}), it is true for the objective of \EDP{}: we show this in Appendix~\ref{sec:non-monotonicity}. This motivates our study of more complex mechanisms.

\section{Approximation Results}\label{sec:approximation}
Previous approaches towards designing truthful, budget feasible mechanisms for \textsc{Knapsack}~\cite{chen} and \textsc{Coverage}~\cite{singer-influence} build upon polynomial-time algorithms that compute an approximation of $OPT$, the optimal value in the full information case. Crucially, to be used in designing a truthful mechanism, such algorithms need also to be \emph{monotone}, in the sense that decreasing any cost $c_i$ leads to an increase  in the estimation of $OPT$; 
the monotonicity property  precludes using traditional approximation algorithms.

In the first part of this section, we address this issue by designing a convex relaxation of \EDP{}, and showing that its solution can be used to approximate $OPT$. The objective of this relaxation is concave and self-concordant \cite{boyd2004convex} and, as such, there exists an algorithm that solves this relaxed problem with arbitrary accuracy in polynomial time. Unfortunately, the output of this algorithm may not necessarily be monotone. Nevertheless, in the second part of this section, we show that a  solver of the relaxed problem can be used to construct a solver that is ``almost'' monotone. In Section~\ref{sec:main}, we show that this algorithm can be used to design a $\delta$-truthful mechanism for \EDP.





\subsection{A Convex Relaxation of \EDP}\label{sec:concave}

A classical way of relaxing combinatorial optimization problems is 
\emph{relaxing by expectation}, using the so-called \emph{multi-linear}
extension of the objective function $V$ (see, \emph{e.g.}, \cite{calinescu2007maximizing,vondrak2008optimal,dughmi2011convex}).
This is because this extension can yield approximation guarantees for a wide class of combinatorial problems through \emph{pipage rounding}, a technique proposed by \citeN{pipage}. Crucially for our purposes, such relaxations in general preserve monotonicity which, as discussed, is required in mechanism design.

Formally, let $P_\mathcal{N}^\lambda$ be a probability distribution over $\mathcal{N}$ parametrized by $\lambda\in [0,1]^n$, where a set $S\subseteq \mathcal{N}$ sampled from  $P_\mathcal{N}^\lambda$ is constructed as follows:  each $i\in \mathcal{N}$ is selected to be in $S$ independently with probability $\lambda_i$, \emph{i.e.},
$    P_\mathcal{N}^\lambda(S) \defeq \prod_{i\in S} \lambda_i
    \prod_{i\in\mathcal{N}\setminus S}( 1 - \lambda_i).$
Then, the \emph{multi-linear} extension $F:[0,1]^n\to\reals$ of $V$ is defined as the
expectation of $V$ under the  distribution $P_\mathcal{N}^\lambda$:
\begin{equation}\label{eq:multi-linear}
    F(\lambda) 
    \defeq \mathbb{E}_{S\sim P_\mathcal{N}^\lambda}\big[V(S)\big]
= \mathbb{E}_{S\sim P_\mathcal{N}^\lambda}\left[ \log\det\left( I_d + \sum_{i\in S} x_i\T{x_i}\right) \right],\quad \lambda\in[0,1]^n.
\end{equation}
Function $F$ is an extension of $V$ to the domain  $[0,1]^n$, as it equals $V$ on integer inputs: $F(\id_S) = V(S)$ for all
$S\subseteq\mathcal{N}$, where $\id_S$ denotes the indicator vector of $S$. 
Contrary to problems such as \textsc{Knapsack}, the
multi-linear extension \eqref{eq:multi-linear} cannot be optimized in
polynomial time for  the value function $V$  we study here, given by \eqref{modified}. Hence, we introduce an extension $L:[0,1]^n\to\reals$ s.t.~
\begin{equation}\label{eq:our-relaxation}
\begin{split}
\quad L(\lambda) &\defeq
\log\det\left(I_d + \sum_{i\in\mathcal{N}} \lambda_i x_i\T{x_i}\right)\\
&=
\log\det\left(\mathbb{E}_{S\sim P_\mathcal{N}^\lambda}\bigg[I_d + \sum_{i\in S}
x_i\T{x_i} \bigg]\right),\quad \lambda\in[0,1]^n
\end{split}
\end{equation}
Note that $L$ also extends $V$, and follows naturally from the multi-linear extension by swapping the
expectation and $\log \det$ in \eqref{eq:multi-linear}. Crucially, it is \emph{strictly concave} on $[0,1]^n$, a fact that  we exploit in the next section to maximize $L$ subject to the budget constraint in polynomial time.

Our first technical lemma relates the concave extension $L$ to the multi-linear extension $F$:
\begin{lemma}\label{lemma:relaxation-ratio}
For all $\lambda\in[0,1]^{n},$
     $           \frac{1}{2}
        \,L(\lambda)\leq
        F(\lambda)\leq L(\lambda)$.
\end{lemma}
The proof of this lemma can be found in Appendix~\ref{proofofrelaxation-ratio}. In short,  exploiting the concavity of the $\log\det$ function over the set of positive semi-definite matrices, we first  bound the ratio of all partial derivatives of $F$ and $L$. We then show that the bound on the ratio of the derivatives also implies a bound on the ratio $F/L$.

Armed with this result, we subsequently use  pipage rounding to show that any $\lambda$ that maximizes the multi-linear extension $F$ can be rounded to an ``almost'' integral solution. More specifically, given a set of costs $c\in \reals^n_+$, we say that a $\lambda\in [0,1]^n$ is feasible if it belongs to the set 
\begin{align}\dom_c =\{\lambda \in [0,1]^n: \sum_{i\in \mathcal{N}} c_i\lambda_i\leq B\}.\label{fdom}\end{align} Then, the following lemma holds:
\begin{lemma}[Rounding]\label{lemma:rounding}
    For any feasible $\lambda\in \dom_c$, there exists a feasible
    $\bar{\lambda}\in \dom_c$ such that (a) $F(\lambda)\leq F(\bar{\lambda})$, and (b) at most one of the 
    coordinates  of $\bar{\lambda}$ is fractional. 
\end{lemma}
The proof of this lemma is in Appendix \ref{proofoflemmarounding}, and follows the main steps of the pipage rounding method of \citeN{pipage}. 
Together, Lemma~\ref{lemma:relaxation-ratio} and Lemma~\ref{lemma:rounding} imply that $OPT$, the optimal value of \EDP, can be approximated by solving the following convex optimization problem:
\begin{align}\tag{$P_c$}\label{eq:primal}
\begin{split}  \text{Maximize:} &\qquad L(\lambda)\\
\text{subject to:} & \qquad\lambda \in \dom_c
\end{split}
\end{align}
In particular, for $L_c^*\defeq \max_{\lambda\in \dom_c} L(\lambda)$ the optimal value of \eqref{eq:primal}, the following holds:
\begin{proposition}\label{prop:relaxation}
$OPT\leq L^*_c \leq 2 OPT + 2\max_{i\in\mathcal{N}}V(i)$.
\end{proposition}
The proof of this proposition can be found in Appendix~\ref{proofofproprelaxation}. As we discuss in the next section, $L^*_c$ can be computed by a poly-time algorithm at arbitrary accuracy. However, the outcome of this computation may not necessarily be monotone; we address this by converting this poly-time estimator of $L^*_c$ to one that is ``almost'' monotone.


\subsection{Polynomial-Time, Almost-Monotone Approximation}\label{sec:monotonicity}
 The $\log\det$ objective function of \eqref{eq:primal} is strictly concave and \emph{self-concordant} \cite{boyd2004convex}. The maximization of a concave, self-concordant function subject to a set of linear constraints can be performed through the \emph{barrier method} (see, \emph{e.g.}, \cite{boyd2004convex} Section 11.5.5 for  general self-concordant optimization as well as \cite{vandenberghe1998determinant} for a detailed treatment of the $\log\det$ objective). The performance of the barrier method is summarized in our case by the following lemma:
\begin{lemma}[\citeN{boyd2004convex}]\label{lemma:barrier}
For any $\varepsilon>0$, the barrier method computes  an 
approximation $\hat{L}^*_c$ that is $\varepsilon$-accurate, \emph{i.e.}, it satisfies $|\hat L^*_c- L^*_c|\leq \varepsilon$, in time $O\left(poly(n,d,\log\log\varepsilon^{-1})\right)$. The same guarantees apply when maximizing $L$ subject to an arbitrary set of $O(n)$ linear constraints.
\end{lemma}

 Clearly, the optimal value $L^*_c$ of \eqref{eq:primal} is monotone in $c$: formally, for any two $c,c'\in \reals_+^n$ s.t.~$c\leq c'$ coordinate-wise, $\dom_{c'}\subseteq \dom_c$ and thus $L^*_c\geq L^*_{c'}$. Hence, the map $c\mapsto L^*_c$ is non-increasing. Unfortunately, the same is not true for the output $\hat{L}_c^*$ of the barrier method: there is no guarantee that the $\epsilon$-accurate approximation $\hat{L}^*_c$ exhibits any kind of monotonicity.

Nevertheless, we prove that it is possible to use the barrier method to construct an approximation of $L^*_{c}$ that is ``almost'' monotone. More specifically, given $\delta>0$,  we say that $f:\reals^n\to\reals$ is
\emph{$\delta$-decreasing}  if
 $ f(x) \geq  f(x+\mu e_i)$, 
    for all $i\in \mathcal{N},x\in\reals^n, \mu\geq\delta,$
where $e_i$ is the $i$-th canonical basis vector of $\reals^n$.
In other words, $f$ is $\delta$-decreasing if increasing any coordinate by $\delta$ or more at input $x$ ensures that the output will be at most $f(x)$.

Our next technical result establishes that, using the barrier method, it is possible to construct an algorithm that computes $L^*_c$ at arbitrary accuracy in polynomial time \emph{and} is $\delta$-decreasing. We achieve this by restricting the optimization over a subset of $\dom_c$ at which the concave relaxation $L$ is ``sufficiently'' concave. Formally, for $\alpha\geq 0$ let $$\textstyle\dom_{c,\alpha} \defeq \{\lambda \in [\alpha,1]^n: \sum_{i\in \mathcal{N}}c_i\lambda_i \leq B\}\subseteq  \dom_c . $$ 
Note that $\dom_c=\dom_{c,0}.$ Consider the following perturbation of the concave relaxation \eqref{eq:primal}:
\begin{align}\tag{$P_{c,\alpha}$}\label{eq:perturbed-primal}
\begin{split}  \text{Maximize:} &\qquad L(\lambda)\\
\text{subject to:} & \qquad\lambda \in \dom_{c,\alpha}
\end{split}
\end{align}

\begin{algorithm}[t]
    \caption{}\label{alg:monotone}
    \begin{algorithmic}[1]
	\Require{ $B\in \reals_+$, $c\in[0,B]^n$, $\delta\in (0,1]$, $\epsilon\in (0,1]$ }
        \State $\alpha \gets \varepsilon (\delta/B+n^2)^{-1}$ 
        \State Use the barrier method to solve \eqref{eq:perturbed-primal} with
        accuracy $\varepsilon'=\frac{1}{2^{n+1}B}\alpha\delta b$; denote the output by $\hat{L}^*_{c,\alpha}$
	\State \textbf{return} $\hat{L}^*_{c,\alpha}$
    \end{algorithmic}
\end{algorithm}

Our construction of a $\delta$-decreasing, $\varepsilon$-accurate approximator of $L_c^*$ proceeds as follows: first, it computes an appropriately selected lower bound $\alpha$; using this bound, it solves the perturbed problem \eqref{eq:perturbed-primal} using the barrier method, also at an appropriately selected accuracy $\varepsilon'$, obtaining thus a $\varepsilon'$-accurate approximation of $L^*_{c,\alpha}\defeq \max_{\lambda\in \dom_{c,\alpha}} L(\lambda)$ . The corresponding output is returned as an approximation of $L^*_c$. A summary of the algorithm and the specific choices of $\alpha$ and $\varepsilon'$ are given in Algorithm~\ref{alg:monotone}. The following proposition, which is proved in Appendix~\ref{proofofpropmonotonicity}, establishes that this algorithm has both properties we desire:
\begin{proposition}\label{prop:monotonicity}
    For any $\delta\in(0,1]$ and any $\varepsilon\in(0,1]$,
    Algorithm~\ref{alg:monotone} computes a $\delta$-decreasing,
    $\varepsilon$-accurate approximation of $L^*_c$. The running time of the
    algorithm is $O\big(poly(n, d, \log\log\frac{B}{b\varepsilon\delta})\big)$.
\end{proposition}
We note that the execution of the barrier method on the restricted set $\dom_{c,\alpha}$ is necessary. The algorithm's output when executed over the entire domain may not necessarily be $\delta$-decreasing, even when the approximation accuracy is small. This is because costs become saturated when the optimal $\lambda\in \dom_c$ lies at the boundary: increasing them has no effect on the objective. Forcing the optimization to happen ``off'' the boundary ensures that this does not occur, while taking $\alpha$ to be small ensures that this perturbation does not cost much in terms of approximation accuracy.

\section{Mechanism for \SEDP{}}\label{sec:mechanism}
\label{sec:main}

In this section we use the $\delta$-decreasing, $\epsilon$-accurate algorithm solving the convex optimization problem \eqref{eq:primal} to design a mechanism for \SEDP. The construction follows a methodology proposed in \cite{singer-mechanisms} and employed by \citeN{chen} and \citeN{singer-influence} to construct deterministic, truthful mechanisms for \textsc{Knapsack} and \textsc{Coverage} respectively. We briefly outline this below (see also Algorithm~\ref{mechanism} in Appendix~\ref{sec:proofofmainthm} for a detailed description).


Recall from Section~\ref{sec:fullinfo} that $i^*\defeq \arg\max_{i\in \mathcal{N}} V(\{i\})$ is the element of maximum value, and $S_G$ is a set constructed greedily, by selecting elements of maximum marginal value per cost. The general framework used by \citeN{chen}  and by \citeN{singer-influence} for the \textsc{Knapsack}   and \textsc{Coverage} value functions contructs an allocation as follows. First, a polynomial-time, monotone approximation of $OPT$ is computed over all elements excluding $i^*$. The outcome of this approximation is compared to $V(\{i^*\})$: if it exceeds $V(\{i^*\})$, then the mechanism constructs an allocation $S_G$ greedily; otherwise, the only item allocated is the singleton $\{i^*\}$.  Provided that the approximation used is within a constant from $OPT$, the above allocation can be shown to also yield a constant approximation to $OPT$. Furthermore, using Myerson's Theorem~\cite{myerson}, it can be shown that this allocation combined with \emph{threshold payments} (see Lemma~\ref{thm:myerson-variant} below) constitute a truthful mechanism. 

The approximation algorithms used in \cite{chen,singer-influence} are LP relaxations, and thus their outputs are monotone and can be computed exactly in polynomial time. We show that the convex relaxation \eqref{eq:primal}, which can be solved by an $\epsilon$-accurate, $\delta$-decreasing algorithm, can be used to construct a $\delta$-truthful, constant approximation mechanism, by being incorporated in the same framework.

To obtain this result, we use the following modified version of Myerson's theorem \cite{myerson}, whose proof we provide in Appendix~\ref{sec:myerson}.

\begin{lemma}\label{thm:myerson-variant}
 A normalized mechanism $\mathcal{M} = (S,p)$ for a single parameter auction is
$\delta$-truthful if:
(a) $S$ is $\delta$-monotone, \emph{i.e.}, for any agent $i$ and $c_i' \leq
c_i-\delta$, for any
fixed costs $c_{-i}$ of agents in $\mathcal{N}\setminus\{i\}$, $i\in S(c_i,
c_{-i})$ implies $i\in S(c_i', c_{-i})$, and (b)
 agents are paid \emph{threshold payments}, \emph{i.e.}, for all $i\in S(c)$, $p_i(c)=\inf\{c_i': i\in S(c_i', c_{-i})\}$.
\end{lemma}
Lemma~\ref{thm:myerson-variant} allows us to incorporate our relaxation in the above framework, yielding the following theorem:
\begin{theorem}\label{thm:main}
    \sloppy
    For any $\delta\in(0,1]$, and any $\epsilon\in (0,1]$,  there exists a $\delta$-truthful, individually rational
    and budget feasible mechanim for \EDP{} that runs in time
    $O\big(poly(n, d, \log\log\frac{B}{b\varepsilon\delta})\big)$
    and allocates
    a set $S^*$ such that
       $ OPT
         \leq \frac{10e-3 + \sqrt{64e^2-24e + 9}}{2(e-1)} V(S^*)+
        \varepsilon
         \simeq 12.98V(S^*) + \varepsilon.$
\end{theorem}

\fussy
The proof of the theorem, as well as our proposed mechanism, can be found in Appendix~\ref{sec:proofofmainthm}.
In addition, we prove the following simple lower bound, proved in Appendix~\ref{proofoflowerbound}.

\begin{theorem}\label{thm:lowerbound}
There is no $2$-approximate, truthful, budget feasible, individually rational
mechanism for EDP. 
\end{theorem}

\section{Conclusions}\label{sec:concl}
We have proposed a convex relaxation for \EDP, and showed that it can be used to design a $\delta$-truthful, constant approximation mechanism that runs in polynomial time. Our objective function, commonly known as the Bayes $D$-optimality criterion, is motivated by linear regression, and in particular captures the information gain when experiments are used to learn a linear model. 

A natural question to ask is to what extent the results we present here
generalize to other machine learning tasks beyond linear regression. We outline
a path in pursuing such generalizations in Appendix~\ref{sec:ext}. In
particular, although the information gain is not generally a submodular
function, we show that for a wide class of models, in which experiments
outcomes are perturbed by independent noise, the information gain indeed
exhibits submodularity. Several important  learning tasks fall under this
category, including generalized linear regression, logistic regression,
\emph{etc.} In light of this, it would be interesting to investigate whether
our convex relaxation approach generalizes to other learning tasks in this
broader class.

The literature on experimental design includes several other optimality
criteria~\cite{pukelsheim2006optimal,atkinson2007optimum}. Our convex
relaxation \eqref{eq:our-relaxation} involved swapping the $\log\det$
scalarization with the expectation appearing in the multi-linear extension
\eqref{eq:multi-linear}. The same swap is known to yield concave objectives for
several other optimality criteria, even when the latter are not submodular
(see, \emph{e.g.}, \citeN{boyd2004convex}). Exploiting the convexity of such
relaxations to design budget feasible mechanisms is an additional open problem
of interest.


\section*{Acknowledgments}
We thank Francis Bach for our helpful discussions on approximate solutions of
convex optimization problems, and Yaron Singer for his comments and
suggestions, and for insights into budget feasible mechanisms.

\bibliographystyle{abbrvnat}
\bibliography{notes}

\begin{thebibliography}{34}
\providecommand{\natexlab}[1]{#1}
\providecommand{\url}[1]{\texttt{#1}}
\expandafter\ifx\csname urlstyle\endcsname\relax
  \providecommand{\doi}[1]{doi: #1}\else
  \providecommand{\doi}{doi: \begingroup \urlstyle{rm}\Url}\fi

\bibitem[Ageev and Sviridenko(2004)]{pipage}
A.~A. Ageev and M.~Sviridenko.
\newblock Pipage rounding: A new method of constructing algorithms with proven
  performance guarantee.
\newblock \emph{J. Comb. Optim.}, 8\penalty0 (3):\penalty0 307--328, 2004.

\bibitem[Akritas et~al.(1996)Akritas, Akritas, and Malaschonok]{sylvester}
A.~G. Akritas, E.~K. Akritas, and G.~I. Malaschonok.
\newblock Various proofs of {Sylvester's} (determinant) identity.
\newblock \emph{Mathematics and Computers in Simulation}, 42\penalty0
  (4–6):\penalty0 585 -- 593, 1996.

\bibitem[Archer et~al.(2004)Archer, Papadimitriou, Talwar, and
  Tardos]{archer-approximate}
A.~Archer, C.~Papadimitriou, K.~Talwar, and E.~Tardos.
\newblock An approximate truthful mechanism for combinatorial auctions with
  single parameter agents.
\newblock \emph{Internet Mathematics}, 1\penalty0 (2):\penalty0 129–150,
  2004.

\bibitem[Atkinson et~al.(2007)Atkinson, Donev, and Tobias]{atkinson2007optimum}
A.~Atkinson, A.~Donev, and R.~Tobias.
\newblock \emph{Optimum experimental designs, with SAS}.
\newblock Oxford University Press New York, 2007.

\bibitem[Badanidiyuru et~al.(2012)Badanidiyuru, Kleinberg, and
  Singer]{singerposted}
A.~Badanidiyuru, R.~Kleinberg, and Y.~Singer.
\newblock Learning on a budget: posted price mechanisms for online procurement.
\newblock In \emph{EC}, 2012.

\bibitem[Bei et~al.(2012)Bei, Chen, Gravin, and Lu]{bei2012budget}
X.~Bei, N.~Chen, N.~Gravin, and P.~Lu.
\newblock Budget feasible mechanism design: from prior-free to bayesian.
\newblock In \emph{STOC}, 2012.

\bibitem[Boyd and Vandenberghe(2004)]{boyd2004convex}
S.~Boyd and L.~Vandenberghe.
\newblock \emph{Convex Optimization}.
\newblock Cambridge University Press, 2004.

\bibitem[Briest et~al.(2005)Briest, Krysta, and
  V\"ocking]{briest-approximation}
P.~Briest, P.~Krysta, and B.~V\"ocking.
\newblock Approximation techniques for utilitarian mechanism design.
\newblock In \emph{Proceedings of the thirty-seventh annual {ACM} symposium on
  Theory of computing}, page 39–48, 2005.

\bibitem[Calinescu et~al.(2007)Calinescu, Chekuri, P{\'a}l, and
  Vondr{\'a}k]{calinescu2007maximizing}
G.~Calinescu, C.~Chekuri, M.~P{\'a}l, and J.~Vondr{\'a}k.
\newblock Maximizing a submodular set function subject to a matroid constraint.
\newblock In \emph{Integer programming and combinatorial optimization}, pages
  182--196. Springer, 2007.

\bibitem[Chaloner and Verdinelli(1995)]{chaloner1995bayesian}
K.~Chaloner and I.~Verdinelli.
\newblock Bayesian experimental design: A review.
\newblock \emph{Statistical Science}, pages 273--304, 1995.

\bibitem[Chen et~al.(2011)Chen, Gravin, and Lu]{chen}
N.~Chen, N.~Gravin, and P.~Lu.
\newblock On the approximability of budget feasible mechanisms.
\newblock In \emph{SODA}, 2011.

\bibitem[Dobzinski et~al.(2011)Dobzinski, Papadimitriou, and
  Singer]{dobz2011-mechanisms}
S.~Dobzinski, C.~H. Papadimitriou, and Y.~Singer.
\newblock Mechanisms for complement-free procurement.
\newblock In \emph{ACM EC}, 2011.

\bibitem[Dughmi(2011)]{dughmi-truthful}
S.~Dughmi.
\newblock A truthful randomized mechanism for combinatorial public projects via
  convex optimization.
\newblock In \emph{EC}, 2011.

\bibitem[Dughmi et~al.(2011)Dughmi, Roughgarden, and Yan]{dughmi2011convex}
S.~Dughmi, T.~Roughgarden, and Q.~Yan.
\newblock From convex optimization to randomized mechanisms: toward optimal
  combinatorial auctions.
\newblock In \emph{STOC}, 2011.

\bibitem[Friedman et~al.(2001)Friedman, Hastie, and Tibshirani]{hastie}
J.~Friedman, T.~Hastie, and R.~Tibshirani.
\newblock \emph{The elements of statistical learning}, volume~1.
\newblock Springer Series in Statistics, 2001.

\bibitem[Ginebra(2007)]{ginebra2007measure}
J.~Ginebra.
\newblock On the measure of the information in a statistical experiment.
\newblock \emph{Bayesian Analysis}, 2\penalty0 (1):\penalty0 167--211, 2007.

\bibitem[Kearns et~al.(2012)Kearns, Pai, Roth, and Ullman]{kearns2012}
M.~Kearns, M.~M. Pai, A.~Roth, and J.~Ullman.
\newblock Private equilibrium release, large games, and no-regret learning,
  2012.
\newblock \url{http://arxiv.org/abs/1207.4084v1}.

\bibitem[Krause and Guestrin(2005{\natexlab{a}})]{krause-submodular}
A.~Krause and C.~Guestrin.
\newblock A note on the budgeted maximization of submodular functions.
\newblock Technical Report CMU-CALD-05-103, CMU, 2005{\natexlab{a}}.

\bibitem[Krause and Guestrin(2005{\natexlab{b}})]{krause2005near}
A.~Krause and C.~Guestrin.
\newblock Near-optimal nonmyopic value of information in graphical models.
\newblock In \emph{UAI}, 2005{\natexlab{b}}.

\bibitem[Lavi and Swamy(2011)]{lavi-truthful}
R.~Lavi and C.~Swamy.
\newblock Truthful and near-optimal mechanism design via linear programming.
\newblock \emph{Journal of the ACM}, 58\penalty0 (6):\penalty0 25, 2011.

\bibitem[Le~Cam(1996)]{le1996comparison}
L.~Le~Cam.
\newblock Comparison of experiments: a short review.
\newblock \emph{Lecture Notes-Monograph Series}, pages 127--138, 1996.

\bibitem[McSherry and Talwar(2007)]{mcsherrytalwar}
F.~McSherry and K.~Talwar.
\newblock Mechanism design via differential privacy.
\newblock In \emph{FOCS}, 2007.

\bibitem[Myerson(1981)]{myerson}
R.~Myerson.
\newblock Optimal auction design.
\newblock \emph{Mathematics of operations research}, 6\penalty0 (1):\penalty0
  58--73, 1981.

\bibitem[Nissim et~al.(2012)Nissim, Smorodinsky, and
  Tennenholtz]{approximatemechanismdesign}
K.~Nissim, R.~Smorodinsky, and M.~Tennenholtz.
\newblock Approximately optimal mechanism design via differential privacy.
\newblock In \emph{Innovations in Theoretical Computer Science (ITCS)}, 2012.

\bibitem[Pai and Roth(2013)]{pai2013privacy}
M.~Pai and A.~Roth.
\newblock Privacy and mechanism design.
\newblock \emph{SIGecom Exchanges}, 2013.

\bibitem[Pukelsheim(2006)]{pukelsheim2006optimal}
F.~Pukelsheim.
\newblock \emph{Optimal design of experiments}, volume~50.
\newblock Society for Industrial Mathematics, 2006.

\bibitem[Schummer(2004)]{schummer2004almost}
J.~Schummer.
\newblock Almost-dominant strategy implementation: exchange economies.
\newblock \emph{Games and Economic Behavior}, 48\penalty0 (1):\penalty0
  154--170, 2004.

\bibitem[Sherman and Morrison(1950)]{sm}
J.~Sherman and W.~Morrison.
\newblock Adjustment of an inverse matrix corresponding to a change in one
  element of a given matrix.
\newblock \emph{The Annals of Mathematical Statistics}, 21\penalty0
  (1):\penalty0 124--127, 1950.

\bibitem[Singer(2010)]{singer-mechanisms}
Y.~Singer.
\newblock Budget feasible mechanisms.
\newblock In \emph{FOCS}, 2010.

\bibitem[Singer(2012)]{singer-influence}
Y.~Singer.
\newblock How to win friends and influence people, truthfully: influence
  maximization mechanisms for social networks.
\newblock In \emph{WSDM}, 2012.

\bibitem[Sviridenko(2004)]{sviridenko-submodular}
M.~Sviridenko.
\newblock A note on maximizing a submodular set function subject to a knapsack
  constraint.
\newblock \emph{Oper. Res. Lett.}, 32\penalty0 (1):\penalty0 41--43, 2004.

\bibitem[Valiant(1984)]{valiant}
L.~G. Valiant.
\newblock A theory of the learnable.
\newblock In R.~A. DeMillo, editor, \emph{STOC}, pages 436--445. ACM, 1984.

\bibitem[Vandenberghe et~al.(1998)Vandenberghe, Boyd, and
  Wu]{vandenberghe1998determinant}
L.~Vandenberghe, S.~Boyd, and S.~Wu.
\newblock Determinant maximization with linear matrix inequality constraints.
\newblock \emph{SIAM journal on matrix analysis and applications}, 19\penalty0
  (2):\penalty0 499--533, 1998.

\bibitem[Vondrak(2008)]{vondrak2008optimal}
J.~Vondrak.
\newblock Optimal approximation for the submodular welfare problem in the value
  oracle model.
\newblock In \emph{Proceedings of the 40th annual ACM symposium on Theory of
  computing}, pages 67--74. ACM, 2008.

\end{thebibliography}
\appendix
\section{Properties of the Value Function $V$} \label{app:properties}
For the sake of concreteness, we prove below the positivity, monotonicity, and submodularity of $V(S) = \log\det(I_d+X_S^TX_S)$ from basic principles. We note however that these properties hold more generally for the information gain under a wider class of models than the linear model with Gaussian noise and prior that we study here: we discuss this in more detail in Appendix~\ref{sec:ext}.

For two symmetric matrices $A$ and $B$, we write $A\succ B$ ($A\succeq B$) if
$A-B$ is positive definite (positive semi-definite).  This order allows us to
define the notion of a \emph{decreasing} as well as  \emph{convex} matrix
function, similarly to their real counterparts. With this definition, matrix
inversion is decreasing and convex over symmetric positive definite
matrices (see Example 3.48 p. 110 in \cite{boyd2004convex}).

Recall that the determinant of a matrix equals the product of its eigenvalues. The positivity of $V(S)$ follows from the fact that $X_S^TX_S$ is positive semi-definite and, as such $I_d+X_S^TX_S\succeq I_d$, so all its eigenvalues are larger than or equal to one, and they are all one if $S=\emptyset$. The marginal contribution of item $i\in\mathcal{N}$ to set
    $S\subseteq \mathcal{N}$ can be written as 
\begin{align}
V(S\cup \{i\}) - V(S)& =   \frac{1}{2}\log\det(I_d 
    + \T{X_S}X_S + x_i\T{x_i})
    - \frac{1}{2}\log\det(I_d + \T{X_S}X_S)\nonumber\\
    &  = \frac{1}{2}\log\det(I_d + x_i\T{x_i}(I_d +
\T{X_S}X_S)^{-1})
 = \frac{1}{2}\log(1 + \T{x_i}A(S)^{-1}x_i)\label{eq:marginal_contrib}
\end{align}
where $A(S) \defeq I_d+ \T{X_S}X_S$, and the last equality follows from the
Sylvester's determinant identity~\cite{sylvester}. Monotonicity therefore follows from the fact that $A(S)^{-1}$ is positive semidefinite. Finally, since the inverse is decreasing over positive definite matrices, we have
\begin{gather}
        \forall S\subseteq\mathcal{N},\quad A(S)^{-1} \succeq A(S\cup\{i\})^{-1}. \label{eq:inverse}
\end{gather}
and submodularity also follows, as a function is submodular if and only if the marginal contributions are non-increasing in $S$. \qed

\section{Proofs of Statements in Section~\ref{sec:concave}}
\subsection{Proof of Lemma~\ref{lemma:relaxation-ratio}}\label{proofofrelaxation-ratio}
    The bound $F(\lambda)\leq L(\lambda)$ follows by the concavity of the $\log\det$ function and Jensen's inequality.
    To show the lower bound, 
    we first prove that $\frac{1}{2}$ is a lower bound of the ratio $\partial_i
    F(\lambda)/\partial_i L(\lambda)$, where we use
    $\partial_i\, \cdot$ as a shorthand for  $\frac{\partial}{\partial \lambda_i}$, the partial derivative  with respect to the
    $i$-th variable. 

   Let us start by computing the partial derivatives of $F$ and
    $L$ with respect to the $i$-th component. 
    Observe that
    \begin{displaymath}
        \partial_i P_\mathcal{N}^\lambda(S) = \left\{
            \begin{aligned}
                & P_{\mathcal{N}\setminus\{i\}}^\lambda(S\setminus\{i\})\;\textrm{if}\;
                i\in S, \\
                & - P_{\mathcal{N}\setminus\{i\}}^\lambda(S)\;\textrm{if}\;
                i\in \mathcal{N}\setminus S. \\
            \end{aligned}\right.
    \end{displaymath}
    Hence,
    \begin{displaymath}
        \partial_i F(\lambda) =
        \sum_{\substack{S\subseteq\mathcal{N}\\ i\in S}}
        P_{\mathcal{N}\setminus\{i\}}^\lambda(S\setminus\{i\})V(S)
        - \sum_{\substack{S\subseteq\mathcal{N}\\ i\in \mathcal{N}\setminus S}}
        P_{\mathcal{N}\setminus\{i\}}^\lambda(S)V(S).
    \end{displaymath}
    Now, using that every $S$ such that $i\in S$ can be uniquely written as
    $S'\cup\{i\}$, we can write:
    \begin{displaymath}
        \partial_i F(\lambda) =
        \sum_{\substack{S\subseteq\mathcal{N}\\ i\in\mathcal{N}\setminus S}}
        P_{\mathcal{N}\setminus\{i\}}^\lambda(S)\big(V(S\cup\{i\})
        - V(S)\big).
    \end{displaymath}
Recall from \eqref{eq:marginal_contrib} that the marginal contribution of $i$ to $S$ is given by
$$V(S\cup \{i\}) - V(S) =\frac{1}{2}\log(1 + \T{x_i}A(S)^{-1}x_i), $$
where $A(S) = I_d+ \T{X_S}X_S$.
Using this,
    \begin{displaymath}
        \partial_i F(\lambda) = \frac{1}{2}
        \sum_{\substack{S\subseteq\mathcal{N}\\ i\in\mathcal{N}\setminus S}}
        P_{\mathcal{N}\setminus\{i\}}^\lambda(S)
        \log\Big(1 + \T{x_i}A(S)^{-1}x_i\Big)
    \end{displaymath}
     The computation of the derivative of $L$ uses standard matrix
    calculus: writing $\tilde{A}(\lambda) \defeq I_d+\sum_{i\in
    \mathcal{N}}\lambda_ix_i\T{x_i}$,
    \begin{displaymath}
        \det \tilde{A}(\lambda + h\cdot e_i) = \det\big(\tilde{A}(\lambda)
        + hx_i\T{x_i}\big)
         =\det \tilde{A}(\lambda)\big(1+
        h\T{x_i}\tilde{A}(\lambda)^{-1}x_i\big).
    \end{displaymath}
    Hence,
    \begin{displaymath}
       \log\det\tilde{A}(\lambda + h\cdot e_i)= \log\det\tilde{A}(\lambda)
        + h\T{x_i}\tilde{A}(\lambda)^{-1}x_i + o(h),
    \end{displaymath}
    which implies
    \begin{displaymath}
        \partial_i L(\lambda)
        =\frac{1}{2} \T{x_i}\tilde{A}(\lambda)^{-1}x_i.
    \end{displaymath}
Recall from \eqref{eq:inverse} that the monotonicity of the matrix inverse over positive definite matrices implies 
\begin{gather*}
        \forall S\subseteq\mathcal{N},\quad A(S)^{-1} \succeq A(S\cup\{i\})^{-1}
\end{gather*}
as $A(S)\preceq A(S\cup\{i\})$. Observe that since $1\leq \lambda_i\leq 1$, 
$P_{\mathcal{N}\setminus\{i\}}^\lambda(S) \geq P_\mathcal{N}^\lambda(S)$ and
$P_{\mathcal{N}\setminus\{i\}}^\lambda(S)\geq P_{\mathcal{N}}^\lambda(S\cup\{i\})$
for all $S\subseteq\mathcal{N}\setminus\{i\}$. Hence,
\begin{align*}
    \partial_i F(\lambda) 
    \geq &\frac{1}{4}
    \sum_{\substack{S\subseteq\mathcal{N}\\ i\in\mathcal{N}\setminus S}}
    P_{\mathcal{N}}^\lambda(S)
    \log\Big(1 + \T{x_i}A(S)^{-1}x_i\Big)\\
    &+\frac{1}{4}
    \sum_{\substack{S\subseteq\mathcal{N}\\ i\in\mathcal{N}\setminus S}}
    P_{\mathcal{N}}^\lambda(S\cup\{i\})
    \log\Big(1 + \T{x_i}A(S\cup\{i\})^{-1}x_i\Big)\\
    \geq &\frac{1}{4}
    \sum_{S\subseteq\mathcal{N}}
    P_\mathcal{N}^\lambda(S)
    \log\Big(1 + \T{x_i}A(S)^{-1}x_i\Big).
\end{align*}
Using that $A(S)\succeq I_d$ we get that $\T{x_i}A(S)^{-1}x_i \leq
\norm{x_i}_2^2 \leq 1$. Moreover, $\log(1+x)\geq x$ for all $x\leq 1$.
Hence,
\begin{displaymath}
    \partial_i F(\lambda) \geq
    \frac{1}{4}
    \T{x_i}\bigg(\sum_{S\subseteq\mathcal{N}}P_\mathcal{N}^\lambda(S)A(S)^{-1}\bigg)x_i.
\end{displaymath}
Finally, using that the inverse is a matrix convex function over symmetric
positive definite matrices (see Appendix~\ref{app:properties}):
\begin{displaymath}
    \partial_i F(\lambda) \geq
    \frac{1}{4}
    \T{x_i}\bigg(\sum_{S\subseteq\mathcal{N}}P_\mathcal{N}^\lambda(S)A(S)\bigg)^{-1}x_i
    = \frac{1}{4}\T{x_i}\tilde{A}(\lambda)^{-1}x_i
    = \frac{1}{2}
    \partial_i L(\lambda).
\end{displaymath}

Having bound the ratio between the partial derivatives, we now bound the ratio
$F(\lambda)/L(\lambda)$ from below. Consider the following three cases.

First, if the minimum is attained as $\lambda$ converges to zero in,
\emph{e.g.}, the $l_2$ norm, by the Taylor approximation, one can write:
\begin{displaymath}
    \frac{F(\lambda)}{L(\lambda)}
    \sim_{\lambda\rightarrow 0}
    \frac{\sum_{i\in \mathcal{N}}\lambda_i\partial_i F(0)}
    {\sum_{i\in\mathcal{N}}\lambda_i\partial_i L(0)}
    \geq \frac{1}{2},
\end{displaymath}
\emph{i.e.}, the ratio $\frac{F(\lambda)}{L(\lambda)}$ is necessarily bounded
from below by 1/2 for small enough $\lambda$.

Second, if the minimum of the ratio $F(\lambda)/L(\lambda)$ is attained at
a vertex of the hypercube $[0,1]^n$ different from 0. $F$ and $L$ being
relaxations of the value function $V$, they are equal to $V$ on the vertices
which are exactly the binary points. Hence, the minimum is equal to 1 in this
case; in particular, it is greater than $1/2$.

Finally, if the minimum is attained at a point $\lambda^*$ with at least one
coordinate belonging to $(0,1)$, let $i$ be one such coordinate and consider
the function $G_i$:
\begin{displaymath}
    G_i: x \mapsto \frac{F}{L}(\lambda_1^*,\ldots,\lambda_{i-1}^*, x,
    \lambda_{i+1}^*, \ldots, \lambda_n^*).
\end{displaymath}
Then this function attains a minimum at $\lambda^*_i\in(0,1)$ and its
derivative is zero at this point. Hence:
\begin{displaymath}
    0 = G_i'(\lambda^*_i) = \partial_i\left(\frac{F}{L}\right)(\lambda^*).
\end{displaymath}
But $\partial_i(F/L)(\lambda^*)=0$ implies that
\begin{displaymath}
    \frac{F(\lambda^*)}{L(\lambda^*)} = \frac{\partial_i
    F(\lambda^*)}{\partial_i L(\lambda^*)}\geq \frac{1}{2}
\end{displaymath}
using the lower bound on the ratio of the partial derivatives. This concludes
the proof of the lemma. \qed


\subsection{Proof of Lemma~\ref{lemma:rounding}}\label{proofoflemmarounding}
    We give a rounding procedure which, given a feasible $\lambda$ with at least
    two fractional components, returns some feasible $\lambda'$ with one fewer fractional
    component such that $F(\lambda) \leq F(\lambda')$.

    Applying this procedure recursively yields the lemma's result.
    Let us consider such a feasible $\lambda$. Let $i$ and $j$ be two
    fractional components of $\lambda$ and let us define the following
    function:
    \begin{displaymath}
        F_\lambda(\varepsilon) = F(\lambda_\varepsilon)
        \quad\textrm{where} \quad
        \lambda_\varepsilon = \lambda + \varepsilon\left(e_i-\frac{c_i}{c_j}e_j\right)
    \end{displaymath}
    It is easy to see that if $\lambda$ is feasible, then:
    \begin{equation}\label{eq:convex-interval}
        \forall\varepsilon\in\Big[\max\Big(-\lambda_i,(\lambda_j-1)\frac{c_j}{c_i}\Big), \min\Big(1-\lambda_i, \lambda_j
        \frac{c_j}{c_i}\Big)\Big],\;
            \lambda_\varepsilon\;\;\textrm{is feasible}
    \end{equation}
    Furthermore, the function $F_\lambda$ is convex; indeed:
    \begin{align*}
        F_\lambda(\varepsilon)
        & = \mathbb{E}_{S'\sim P_{\mathcal{N}\setminus\{i,j\}}^\lambda(S')}\Big[
        (\lambda_i+\varepsilon)\Big(\lambda_j-\varepsilon\frac{c_i}{c_j}\Big)V(S'\cup\{i,j\})\\
        & + (\lambda_i+\varepsilon)\Big(1-\lambda_j+\varepsilon\frac{c_i}{c_j}\Big)V(S'\cup\{i\})
         + (1-\lambda_i-\varepsilon)\Big(\lambda_j-\varepsilon\frac{c_i}{c_j}\Big)V(S'\cup\{j\})\\
         & + (1-\lambda_i-\varepsilon)\Big(1-\lambda_j+\varepsilon\frac{c_i}{c_j}\Big)V(S')\Big]
    \end{align*}
    Thus, $F_\lambda$ is a degree 2 polynomial whose dominant coefficient is:
    \begin{displaymath}
        \frac{c_i}{c_j}\mathbb{E}_{S'\sim
        P_{\mathcal{N}\setminus\{i,j\}}^\lambda(S')}\Big[
            V(S'\cup\{i\})+V(S'\cup\{i\})\\
        -V(S'\cup\{i,j\})-V(S')\Big]
    \end{displaymath}
    which is positive by submodularity of $V$. Hence, the maximum of
    $F_\lambda$ over the interval given in \eqref{eq:convex-interval} is
    attained at one of its limits, at which either the $i$-th or $j$-th component of
    $\lambda_\varepsilon$ becomes integral. \qed
\subsection{Proof of Proposition~\ref{prop:relaxation}}\label{proofofproprelaxation}
The lower bound on $L^*_c$ follows immediately from the fact that $L$ extends $V$ to $[0,1]^n$. For the upper bound, let us consider a feasible point $\lambda^*\in \dom_c$ such that
$L(\lambda^*) = L^*_c$. By applying Lemma~\ref{lemma:relaxation-ratio} and
Lemma~\ref{lemma:rounding} we get a feasible point $\bar{\lambda}$ with at most
one fractional component such that
\begin{equation}\label{eq:e1}
    L(\lambda^*) \leq 2 F(\bar{\lambda}).
\end{equation}
    Let $\lambda_i$ denote the fractional component of $\bar{\lambda}$ and $S$
    denote the set whose indicator vector is $\bar{\lambda} - \lambda_i e_i$.
    By definition of the multi-linear extension $F$:
    \begin{displaymath}
        F(\bar{\lambda}) = (1-\lambda_i)V(S) +\lambda_i V(S\cup\{i\}).
    \end{displaymath}
    By submodularity of $V$, $V(S\cup\{i\})\leq V(S) + V(\{i\})$. Hence,
    \begin{displaymath}
        F(\bar{\lambda}) \leq V(S) + V(i).
    \end{displaymath}
    Note that since $\bar{\lambda}$ is feasible, $S$ is also feasible and
    $V(S)\leq OPT$. Hence,
    \begin{equation}\label{eq:e2}
        F(\bar{\lambda}) \leq  OPT + \max_{i\in\mathcal{N}} V(i).
    \end{equation}
Together, \eqref{eq:e1} and \eqref{eq:e2} imply the proposition.\qed

\section{Proof of Proposition~\ref{prop:monotonicity}}\label{proofofpropmonotonicity}


We proceed by showing that the optimal value of \eqref{eq:perturbed-primal} is close to the
optimal value of \eqref{eq:primal} (Lemma~\ref{lemma:proximity}) while being
well-behaved with respect to changes of the cost
(Lemma~\ref{lemma:monotonicity}). These lemmas together imply
Proposition~\ref{prop:monotonicity}.

Note that the choice of $\alpha$ given in Algorithm~\ref{alg:monotone} implies
that $\alpha<\frac{1}{n}$. This in turn implies that the feasible set
$\mathcal{D}_{c, \alpha}$ of \eqref{eq:perturbed-primal} is non-empty: it
contains the strictly feasible point $\lambda=(\frac{1}{n},\ldots,\frac{1}{n})$.

\begin{lemma}\label{lemma:derivative-bounds}
    Let $\partial_i L(\lambda)$ denote the $i$-th derivative of $L$, for $i\in\{1,\ldots, n\}$, then:
    \begin{displaymath}
        \forall\lambda\in[0, 1]^n,\;\frac{b}{2^n} \leq \partial_i L(\lambda) \leq 1
    \end{displaymath}
\end{lemma}

\begin{proof}
    Recall that we had defined:
    \begin{displaymath}
        \tilde{A}(\lambda)\defeq I_d + \sum_{i=1}^n \lambda_i x_i\T{x_i}
        \quad\mathrm{and}\quad
        A(S) \defeq I_d + \sum_{i\in S} x_i\T{x_i}
    \end{displaymath}
    Let us also define $A_k\defeq A(\{x_1,\ldots,x_k\})$.
    We have $\partial_i L(\lambda) = \T{x_i}\tilde{A}(\lambda)^{-1}x_i$. Since
    $\tilde{A}(\lambda)\succeq I_d$, $\partial_i L(\lambda)\leq \T{x_i}x_i \leq 1$, which
    is the right-hand side  of the lemma.
    For the left-hand side, note that $\tilde{A}(\lambda) \preceq A_n$. Hence
    $\partial_iL(\lambda)\geq \T{x_i}A_n^{-1}x_i$.
    Using the Sherman-Morrison formula \cite{sm}, for all $k\geq 1$:
    \begin{displaymath}
        \T{x_i}A_k^{-1} x_i = \T{x_i}A_{k-1}^{-1}x_i 
        - \frac{(\T{x_i}A_{k-1}^{-1}x_k)^2}{1+\T{x_k}A_{k-1}^{-1}x_k}
    \end{displaymath}
    By the Cauchy-Schwarz inequality:
    \begin{displaymath}
        (\T{x_i}A_{k-1}^{-1}x_k)^2 \leq \T{x_i}A_{k-1}^{-1}x_i\;\T{x_k}A_{k-1}^{-1}x_k
    \end{displaymath}
    Hence:
    \begin{displaymath}
        \T{x_i}A_k^{-1} x_i \geq \T{x_i}A_{k-1}^{-1}x_i 
        - \T{x_i}A_{k-1}^{-1}x_i\frac{\T{x_k}A_{k-1}^{-1}x_k}{1+\T{x_k}A_{k-1}^{-1}x_k}
    \end{displaymath}
    But $\T{x_k}A_{k-1}^{-1}x_k\leq 1$ and $\frac{a}{1+a}\leq \frac{1}{2}$ if
    $0\leq a\leq 1$, so:
    \begin{displaymath}
        \T{x_i}A_{k}^{-1}x_i \geq \T{x_i}A_{k-1}^{-1}x_i
        - \frac{1}{2}\T{x_i}A_{k-1}^{-1}x_i\geq \frac{\T{x_i}A_{k-1}^{-1}x_i}{2}
    \end{displaymath}
    By induction:
    \begin{displaymath}
        \T{x_i}A_n^{-1} x_i \geq \frac{\T{x_i}x_i}{2^n}
    \end{displaymath}
    Using that $\T{x_i}{x_i}\geq b$ concludes the proof of the left-hand side
    of the lemma's inequality.
\end{proof}
Let us introduce the Lagrangian of problem \eqref{eq:perturbed-primal}:

\begin{displaymath}
    \mathcal{L}_{c, \alpha}(\lambda, \mu, \nu, \xi) \defeq L(\lambda) 
    + \T{\mu}(\lambda-\alpha\mathbf{1}) + \T{\nu}(\mathbf{1}-\lambda) + \xi(B-\T{c}\lambda)
\end{displaymath}
so that:
\begin{displaymath}
    L^*_{c,\alpha} = \min_{\mu, \nu, \xi\geq 0}\max_\lambda \mathcal{L}_{c, \alpha}(\lambda, \mu, \nu, \xi)
\end{displaymath}
Similarly, we define $\mathcal{L}_{c}\defeq\mathcal{L}_{c, 0}$ the lagrangian of \eqref{eq:primal}.

Let $\lambda^*$ be primal optimal for \eqref{eq:perturbed-primal}, and $(\mu^*,
\nu^*, \xi^*)$ be dual optimal for the same problem. In addition to primal and
dual feasibility, the Karush-Kuhn-Tucker (KKT) conditions \cite{boyd2004convex} give $\forall i\in\{1, \ldots, n\}$:
\begin{gather*}
    \partial_i L(\lambda^*) + \mu_i^* - \nu_i^* - \xi^* c_i = 0\\
    \mu_i^*(\lambda_i^* - \alpha) = 0\\
    \nu_i^*(1 - \lambda_i^*) = 0
\end{gather*}

\begin{lemma}\label{lemma:proximity}
We have:
\begin{displaymath}
    L^*_c - \alpha n^2\leq L^*_{c,\alpha} \leq L^*_c
\end{displaymath}
In particular, $|L^*_c - L^*_{c,\alpha}| \leq \alpha n^2$.
\end{lemma}

\begin{proof}
    $\alpha\mapsto L^*_{c,\alpha}$ is a decreasing function as it is the
    maximum value of the $L$ function over a set-decreasing domain, which gives
    the rightmost inequality.

    Let $\mu^*, \nu^*, \xi^*$ be dual optimal for $(P_{c, \alpha})$, that is:
    \begin{displaymath}
        L^*_{c,\alpha} = \max_\lambda \mathcal{L}_{c, \alpha}(\lambda, \mu^*, \nu^*, \xi^*)
    \end{displaymath}

    Note that $\mathcal{L}_{c, \alpha}(\lambda, \mu^*, \nu^*, \xi^*)
    = \mathcal{L}_{c}(\lambda, \mu^*, \nu^*, \xi^*)
    - \alpha\T{\mathbf{1}}\mu^*$, and that for any $\lambda$ feasible for
    problem \eqref{eq:primal}, $\mathcal{L}_{c}(\lambda, \mu^*, \nu^*, \xi^*)
    \geq L(\lambda)$. Hence,
    \begin{displaymath}
        L^*_{c,\alpha} \geq L(\lambda) - \alpha\T{\mathbf{1}}\mu^*
    \end{displaymath}
    for any $\lambda$ feasible for \eqref{eq:primal}. In particular, for $\lambda$ primal optimal for $\eqref{eq:primal}$:
    \begin{equation}\label{eq:local-1}
        L^*_{c,\alpha} \geq L^*_c - \alpha\T{\mathbf{1}}\mu^*
    \end{equation}

    Let us denote by the $M$ the support of $\mu^*$, that is $M\defeq
    \{i|\mu_i^* > 0\}$, and by $\lambda^*$ a primal optimal point for
    $\eqref{eq:perturbed-primal}$.  From the KKT conditions we see that:
    \begin{displaymath}
        M \subseteq \{i|\lambda_i^* = \alpha\}
    \end{displaymath}

    Let us first assume that $|M| = 0$, then $\T{\mathbf{1}}\mu^*=0$ and the lemma follows.

    We will now assume that $|M|\geq 1$. In this case $\T{c}\lambda^*
    = B$, otherwise we could increase the coordinates of $\lambda^*$ in $M$,
    which would increase the value of the objective function and contradict the
    optimality of $\lambda^*$. Note also, that $|M|\leq n-1$, otherwise, since
    $\alpha< \frac{1}{n}$, we would have $\T{c}\lambda^*\ < B$, which again
    contradicts the optimality of $\lambda^*$. Let us write:
    \begin{displaymath}
        B = \T{c}\lambda^* = \alpha\sum_{i\in M}c_i + \sum_{i\in \bar{M}}\lambda_i^*c_i
        \leq \alpha |M|B + (n-|M|)\max_{i\in \bar{M}} c_i
    \end{displaymath}
    That is:
    \begin{equation}\label{local-2}
        \max_{i\in\bar{M}} c_i \geq \frac{B - B|M|\alpha}{n-|M|}> \frac{B}{n}
    \end{equation}
    where the last inequality uses again that $\alpha<\frac{1}{n}$. From the
    KKT conditions, we see that for $i\in M$, $\nu_i^* = 0$ and:
    \begin{equation}\label{local-3}
        \mu_i^* = \xi^*c_i - \partial_i L(\lambda^*)\leq \xi^*c_i\leq \xi^*B
    \end{equation}
    since $\partial_i L(\lambda^*)\geq 0$ and $c_i\leq 1$.

    Furthermore, using the KKT conditions again, we have that:
    \begin{equation}\label{local-4}
        \xi^* \leq \inf_{i\in \bar{M}}\frac{\partial_i L(\lambda^*)}{c_i}\leq \inf_{i\in \bar{M}} \frac{1}{c_i}
        = \frac{1}{\max_{i\in\bar{M}} c_i}
    \end{equation}
    where the last inequality uses Lemma~\ref{lemma:derivative-bounds}.
    Combining \eqref{local-2}, \eqref{local-3} and \eqref{local-4}, we get that:
    \begin{displaymath}
        \sum_{i\in M}\mu_i^* \leq |M|\xi^*B \leq n\xi^*B\leq \frac{nB}{\max_{i\in\bar{M}} c_i} \leq n^2
    \end{displaymath}
    This implies that:
    \begin{displaymath}
        \T{\mathbf{1}}\mu^* = \sum_{i=1}^n \mu^*_i = \sum_{i\in M}\mu_i^*\leq n^2
    \end{displaymath}
    which along with \eqref{eq:local-1} proves the lemma.
\end{proof} 

\begin{lemma}\label{lemma:monotonicity}
    If $c'$ = $(c_i', c_{-i})$, with $c_i'\leq c_i - \delta$, we have:
    \begin{displaymath}
        L^*_{c',\alpha} \geq L^*_{c,\alpha} + \frac{\alpha\delta b}{2^nB}
    \end{displaymath}
\end{lemma}

\begin{proof}
    Let $\mu^*, \nu^*, \xi^*$ be dual optimal for $(P_{c', \alpha})$. Noting that:
    \begin{displaymath} 
    \mathcal{L}_{c', \alpha}(\lambda, \mu^*, \nu^*, \xi^*) \geq
    \mathcal{L}_{c, \alpha}(\lambda, \mu^*, \nu^*, \xi^*) + \lambda_i\xi^*\delta,
    \end{displaymath}
    we get similarly to Lemma~\ref{lemma:proximity}:
    \begin{displaymath}
        L^*_{c',\alpha} \geq L(\lambda) + \lambda_i\xi^*\delta
    \end{displaymath}
    for any $\lambda$ feasible for \eqref{eq:perturbed-primal}. In particular, for $\lambda^*$ primal optimal for \eqref{eq:perturbed-primal}:
    \begin{displaymath}
        L^*_{c',\alpha} \geq L^*_{c,\alpha} + \alpha\xi^*\delta
    \end{displaymath}
    since $\lambda_i^*\geq \alpha$.

    Using the KKT conditions for $(P_{c', \alpha})$, we can write:
    \begin{displaymath}
        \xi^* = \inf_{i:\lambda^{'*}_i>\alpha} \frac{\T{x_i}S(\lambda^{'*})^{-1}x_i}{c_i'}
    \end{displaymath}
    with $\lambda^{'*}$ optimal for $(P_{c', \alpha})$. Since $c_i'\leq B$,
    using Lemma~\ref{lemma:derivative-bounds}, we get that $\xi^*\geq
    \frac{b}{2^nB}$, which concludes the proof.
\end{proof}

We are now ready to conclude the  proof of Proposition~\ref{prop:monotonicity}.
Let $\hat{L}^*_{c,\alpha}$ be the approximation computed by
Algorithm~\ref{alg:monotone}.
\begin{enumerate}
    \item using Lemma~\ref{lemma:proximity}:
\begin{displaymath}
        |\hat{L}^*_{c,\alpha} - L^*_c| \leq |\hat{L}^*_{c,\alpha} - L^*_{c,\alpha}| + |L^*_{c,\alpha} - L^*_c|
        \leq \frac{\alpha\delta}{B} + \alpha n^2 = \varepsilon
\end{displaymath}
which proves the $\varepsilon$-accuracy.

\item for the $\delta$-decreasingness, let $c' = (c_i', c_{-i})$ with $c_i'\leq
    c_i-\delta$, then:
\begin{displaymath}
    \hat{L}^*_{c',\alpha} \geq L^*_{c',\alpha} - \frac{\alpha\delta b}{2^{n+1}B} 
                     \geq L^*_{c,\alpha} + \frac{\alpha\delta b}{2^{n+1}B}
    \geq \hat{L}^*_{c,\alpha}
\end{displaymath}
where the first and last inequalities follow from the accuracy of the approximation, and
the inner inequality follows from Lemma~\ref{lemma:monotonicity}.

\item the accuracy of the approximation $\hat{L}^*_{c,\alpha}$ is:
\begin{displaymath}
    \varepsilon' =\frac{\varepsilon\delta b}{2^{n+1}(\delta + n^2B)}
\end{displaymath}

Note that:
\begin{displaymath}
    \log\log (\varepsilon')^{-1} = O\bigg(\log\log\frac{B}{\varepsilon\delta b} + \log n\bigg)
\end{displaymath}
Using Lemma~\ref{lemma:barrier} concludes the proof of the running time.\qed
\end{enumerate}

\section{Budget Feasible Reverse Auction Mechanisms}\label{app:budgetfeasible}
We review in this appendix the formal definition of a budget feasible reverse auction mechanisms, as introduced by \citeN{singer-mechanisms}. We depart from the definitions in \cite{singer-mechanisms} only in considering $\delta$-truthful, rather than truthful, mechanisms.

Given a budget $B$ and a value function $V:2^{\mathcal{N}}\to\reals_+$, a \emph{mechanism} $\mathcal{M} = (S,p)$ comprises (a) an \emph{allocation function} 
$S:\reals_+^n \to 2^\mathcal{N}$ and (b) a \emph{payment function}
$p:\reals_+^n\to \reals_+^n$. 
 Let $s_i(c) = \id_{i\in S(c)}$ be the binary indicator of  $i\in S(c)$. We seek mechanisms that have the following properties \cite{singer-mechanisms}:

\sloppy
\begin{itemize}
\item \emph{Normalization.} Unallocated experiments receive zero payments: $s_i(c)=0\text{ implies }p_i(c)=0.\label{normal}$
\item\emph{Individual Rationality.} Payments for allocated experiments exceed
costs: $p_i(c)\geq c_i\cdot s_i(c).\label{ir}$
\item \emph{No Positive Transfers.} Payments are non-negative: $p_i(c)\geq 0\label{pt}$.
\item \emph{$\delta$-Truthfulness.} Reporting one's true cost is
an \emph{almost-dominant} \cite{schummer2004almost} strategy. Formally, let $c_{-i}$
 be a vector of costs of all agents except $i$. Then, $p_i(c_i,c_{-i})
 - s_i(c_i,c_{-i})\cdot c_i \geq p_i(c_i',c_{-i}) - s_i(c_i',c_{-i})\cdot c_i,
 $ for every $i \in \mathcal{N}$ and every two cost vectors $(c_i,c_{-i})$ and $(c_i',c_{-i})$ such that $|c_i-c_i'|>\delta.$ The mechanism is \emph{truthful} if $\delta=0$.
 \item \emph{Budget Feasibility.} The sum of the payments should not exceed the
    budget constraint, \emph{i.e.} $\sum_{i\in\mathcal{N}} p_i(c) \leq
    B.\label{budget}$

   \item \emph{$(\alpha,\beta)$-approximation.} The value of the allocated set  should not
    be too far from the optimum value of the full information case, as given by  \eqref{eq:non-strategic}. 
Formally, there must exist some $\alpha\geq 1$ and $\beta>0$
    such that $OPT \leq \alpha V(S(p))+\beta$, where     $OPT = \max_{S\subseteq\mathcal{N}} \left\{V(S) \;\mid \;
    \sum_{i\in S}c_i\leq B\right\}$.

    \item \emph{Computational Efficiency.} The allocation and payment function
    should be computable in time polynomial in various parameters. 
\end{itemize}

\fussy
\section{Proof of Lemma~\ref{thm:myerson-variant}}\label{sec:myerson}
Using the notations of Lemma~\ref{thm:myerson-variant}, we want to prove that
if $c_i$ and $c_i'$ are two different costs reported by user $i$ with $|c_i
- c_i'|\geq \delta$, and if $c_{-i}$ is any vector of costs reported by the
other users:
\begin{equation}\label{eq:local-foobar}
    p_i(c_i, c_{-i}) - s_i(c_i, c_{-i})\cdot c_i \geq p_i(c_i', c_{-i})
    - s_i(c_i', c_{-i})\cdot c_i
\end{equation}

We distinguish four cases depending on the value of $s_i(c_i, c_{-i})$ and
$s_i'(c_i', c_{-i})$.

\begin{enumerate}
\item $s_i(c_i, c_{-i})= s_i(c_i', c_{-i})=0$.
Since the mechanism is normalized 
we have $p_i(c_i, c_{-i}) = p_i(c_i', c_{-i})= 0$ and \eqref{eq:local-foobar}
is true.
\item $s_i(c_i', c_{-i}) = s_i(c_i, c_{-i}) = 1$.
Note that $i$ is paid her threshold payment when allocated, and since this
payment does not depend on $i$'s reported cost, \eqref{eq:local-foobar} is true
(and is in fact an equality).
\item $s_i(c_i', c_{-i}) = 0$ and $s_i(c_i, c_{-i}) = 1$.
 We then have $p_i(c_i',
c_{-i}) = 0$ by normalization and \eqref{eq:local-foobar} follows from
individual rationality.
\item $s_i(c_i', c_{-i}) = 1$ and $s_i(c_i, c_{-i}) = 0$.
By $\delta$-decreasingness of $s_i$, $c_i \geq c_i'+\delta$, and $s_i(c_i,
c_{-i}) = 0$ implies that $i$'s threshold payment is less than $c_i$,
\emph{i.e.} $p_i(c_i', c_{-i}) \leq c_i$. This last inequality is equivalent to
\eqref{eq:local-foobar} in this final case. \qed
\end{enumerate}

\section{Proof of Theorem~\ref{thm:main}}\label{sec:proofofmainthm}
\begin{algorithm}[!t]
    \caption{Mechanism for \SEDP{}}\label{mechanism}
    \begin{algorithmic}[1]
	\Require{ $B\in \reals_+$,$c\in[0,B]^n$, $\delta\in (0,1]$, $\epsilon\in (0,1]$ }
    \State $i^* \gets \argmax_{j\in\mathcal{N}}V(j)$
    \State \label{relaxexec}$OPT'_{-i^*} \gets$ using Proposition~\ref{prop:monotonicity},
    compute a $\varepsilon$-accurate, $\delta$-decreasing
    approximation of $$\textstyle L^*_{c_{-i^*}}\defeq\max_{\lambda\in[0,1]^{n}} \{L(\lambda)
                                    \mid \lambda_{i^*}=0,\sum_{i \in \mathcal{N}\setminus\{i^*\}}c_i\lambda_i\leq B\}$$
     \State $C \gets \frac{8e-1 + \sqrt{64e^2-24e + 9}}{2(e-1)}$

        \If{$OPT'_{-i^*} < C \cdot V(i^*)$} \label{c} 
       \State \textbf{return} $\{i^*\}$ 
        \Else
            \State $i \gets \argmax_{1\leq j\leq n}\frac{V(j)}{c_j}$
	    \State $S_G \gets \emptyset$
            \While{$c_i\leq \frac{B}{2}\frac{V(S_G\cup\{i\})-V(S_G)}{V(S_G\cup\{i\})}$}
                \State $S_G \gets S_G\cup\{i\}$ 
 \State $i \gets \argmax_{j\in\mathcal{N}\setminus S_G}
                \frac{V(S_G\cup\{j\})-V(S_G)}{c_j}$
            \EndWhile
            \State \textbf{return} $S_G$
        \EndIf
    \end{algorithmic}
\end{algorithm}

We now present the proof of Theorem~\ref{thm:main}. 
Our mechanism for \EDP{} is composed of
(a) the allocation function presented in Algorithm~\ref{mechanism}, and 
(b) the payment function which pays each allocated agent $i$ her threshold
payment as described in Myerson's Theorem.  In the case where $\{i^*\}$ is the
allocated set, her threshold payment is $B$. 
A closed-form formula for threshold payments when $S_G$ is the allocated set can
be found in~\cite{singer-mechanisms}.

We use the notation $OPT_{-i^*}$ to denote the optimal value of \EDP{} when the maximum value element $i^*$ is  excluded. We also use $OPT'_{-i^*}$ to denote the approximation computed by the $\delta$-decreasing, $\epsilon$-accurate approximation of $L^*_{c_{-i^*}}$, as defined in Algorithm~\ref{mechanism}.

\sloppy
The properties of $\delta$-truthfulness and
individual rationality follow from $\delta$-monotonicity and threshold
payments. $\delta$-monotonicity and budget feasibility follow  similar steps as the
analysis of \citeN{chen}, adapted to account for $\delta$-monotonicity:

\fussy
\begin{lemma}\label{lemma:monotone}
Our mechanism for \EDP{} is $\delta$-monotone and budget feasible.
\end{lemma}

\begin{proof}
    Consider an agent $i$ with cost $c_i$ that is
    selected by the mechanism, and suppose that she reports
    a cost $c_i'\leq c_i-\delta$ while all  other costs stay the same.
    Suppose that when $i$ reports $c_i$, $OPT'_{-i^*} \geq C V(i^*)$; then, as $s_i(c_i,c_{-i})=1$, $i\in S_G$.
     By reporting cost $c_i'$, $i$ may be selected at an earlier iteration of the greedy algorithm.
 Denote by $S_i$
    (resp. $S_i'$) the set to which $i$ is added when reporting cost $c_i$
    (resp. $c_i'$). We have $S_i'\subseteq S_i$; in addition, $S_i'\subseteq S_G'$, the set selected by the greedy algorithm under $(c_i',c_{-i})$; if not, then greedy selection would terminate prior to selecting $i$ also when she reports $c_i$, a contradiction. Moreover, we have
    \begin{align*}
        c_i' & \leq c_i \leq
        \frac{B}{2}\frac{V(S_i\cup\{i\})-V(S_i)}{V(S_i\cup\{i\})}
         \leq \frac{B}{2}\frac{V(S_i'\cup\{i\})-V(S_i')}{V(S_i'\cup\{i\})}
    \end{align*}
    by the monotonicity and submodularity  of $V$. Hence  $i\in S_G'$. By
    $\delta$-decreasingness of
    $OPT'_{-i^*}$, under $c'_i\leq c_i-\delta$ the greedy set is still allocated and $s_i(c_i',c_{-i}) =1$.
    Suppose now that when $i$ reports $c_i$, $OPT'_{-i^*} < C V(i^*)$. Then $s_i(c_i,c_{-i})=1$ iff $i = i^*$. 
    Reporting $c_{i^*}'\leq c_{i^*}$ does not change $V(i^*)$ nor
    $OPT'_{-i^*} \leq C V(i^*)$; thus $s_{i^*}(c_{i^*}',c_{-i^*})=1$, so the mechanism is monotone.

To show budget feasibility, suppose that $OPT'_{-i^*} < C V(i^*)$. Then the mechanism selects $i^*$. Since the bid of $i^*$ does not affect the above condition, the threshold payment of $i^*$ is $B$ and the mechanism is budget feasible.
Suppose  that $OPT'_{-i^*} \geq C V(i^*)$. 
Denote by $S_G$ the set selected by the greedy algorithm, and for $i\in S_G$,  denote by
$S_i$ the subset of the solution set that was selected by the greedy algorithm just prior to the addition of $i$---both sets determined for the present cost vector $c$. 
Then for any submodular function $V$, and for all $i\in S_G$:
\begin{equation}\label{eq:budget}
   \text{if}~c_i'\geq \frac{V(S_i\cup\{i\}) - V(S)}{V(S_G)} B~\text{then}~s_i(c_i',c_{-i})=0
\end{equation}
In other words, if $i$ increases her cost to a value higher than $\frac{V(S_i\cup\{i\}) - V(S)}{V(S_G)}$, she will cease to be in the selected set $S_G$. As a result,
\eqref{eq:budget}
implies that the threshold payment of user $i$ is bounded by the above quantity.
Hence, the total payment is bounded by the telescopic sum:
\begin{displaymath}
    \sum_{i\in S_G} \frac{V(S_i\cup\{i\}) - V(S_i)}{V(S_G)} B = \frac{V(S_G)-V(\emptyset)}{V(S_G)} B=B\qedhere
\end{displaymath}
\end{proof}

The complexity of the mechanism is given by the following lemma.

\begin{lemma}[Complexity]\label{lemma:complexity}
    For any $\varepsilon > 0$ and any $\delta>0$, the complexity of the mechanism  is
    $O\big(poly(n, d, \log\log\frac{B}{b\varepsilon\delta})\big)$
\end{lemma}

\begin{proof}
    The value function $V$ in \eqref{modified} can be computed in time
    $O(\text{poly}(n, d))$ and the mechanism only involves a linear
    number of queries to the function $V$. 
    By Proposition~\ref{prop:monotonicity}, line \ref{relaxexec} of Algorithm~\ref{mechanism}
    can be computed in time
    $O(\text{poly}(n, d, \log\log \frac{B}{b\varepsilon\delta}))$. Hence the allocation
    function's complexity is as stated. 
\junk{
    Using Singer's characterization of the threshold payments
    \cite{singer-mechanisms}, one can verify that they can be computed in time
    $O(\text{poly}(n, d))$.
    }
\end{proof}

Finally, we prove the approximation ratio of the mechanism.
We use the following lemma from \cite{chen} which bounds $OPT$ in terms of
the value of $S_G$, as computed in Algorithm \ref{mechanism}, and $i^*$, the
element of maximum value.

\sloppy
\begin{lemma}[\cite{chen}]\label{lemma:greedy-bound}
Let $S_G$ be the set computed in Algorithm \ref{mechanism} and let 
$i^*=\argmax_{i\in\mathcal{N}} V(\{i\})$. We have:
\begin{displaymath}
OPT \leq \frac{e}{e-1}\big( 3 V(S_G) + 2 V(i^*)\big).
\end{displaymath}
\end{lemma}

\fussy
Using Proposition~\ref{prop:relaxation} and Lemma~\ref{lemma:greedy-bound} we can complete the proof of
Theorem~\ref{thm:main} by showing that, for any $\varepsilon > 0$, if
$OPT_{-i}'$, the optimal value of $L$ when $i^*$ is excluded from
$\mathcal{N}$, has been computed to a precision $\varepsilon$, then the set
$S^*$ allocated by the mechanism is such that:
\begin{equation} \label{approxbound}
OPT
\leq \frac{10e\!-\!3 + \sqrt{64e^2\!-\!24e\!+\!9}}{2(e\!-\!1)} V(S^*)\!
+ \! \varepsilon .
\end{equation}
To see this, let $L^*_{c_{-i^*}}$ be the  maximum value of $L$ subject to
$\lambda_{i^*}=0$, $\sum_{i\in \mathcal{N}\setminus{i^*}}c_i\leq B$. From line
\ref{relaxexec} of Algorithm~\ref{mechanism}, we have
$L^*_{c_{-i^*}}-\varepsilon\leq OPT_{-i^*}' \leq L^*_{c_{-i^*}}+\varepsilon$.

If the condition on line \ref{c} of the algorithm holds then, from the lower bound in Proposition~\ref{prop:relaxation},
\begin{displaymath}
    V(i^*) \geq \frac{1}{C}L^*_{c_{-i^*}}-\frac{\varepsilon}{C} \geq
    \frac{1}{C}OPT_{-i^*} -\frac{\varepsilon}{C}.
\end{displaymath}
Also, $OPT \leq OPT_{-i^*} + V(i^*)$,
hence,
\begin{equation}\label{eq:bound1}
    OPT\leq (1+C)V(i^*) + \varepsilon.
\end{equation}
If the condition on line \ref{c} does not hold, by observing that $L^*_{c_{-i^*}}\leq L^*_c$ and
the upper bound of Proposition~\ref{prop:relaxation}, we get
\begin{displaymath}
    V(i^*)\leq \frac{1}{C}L^*_{c_{-i^*}} + \frac{\varepsilon}{C}
    \leq \frac{1}{C} \big(2 OPT + 2 V(i^*)\big) + \frac{\varepsilon}{C}.
\end{displaymath}
Applying Lemma~\ref{lemma:greedy-bound},
\begin{displaymath}
    V(i^*) \leq \frac{1}{C}\left(\frac{2e}{e-1}\big(3 V(S_G)
    + 2 V(i^*)\big) + 2 V(i^*)\right) + \frac{\varepsilon}{C}.
\end{displaymath}
Note that $C$ satifies $C(e-1) -6e  +2 > 0$, hence
\begin{align*}
    V(i^*) \leq \frac{6e}{C(e-1)- 6e + 2} V(S_G) 
    + \frac{(e-1)\varepsilon}{C(e-1)- 6e + 2}.
\end{align*}
Finally, using Lemma~\ref{lemma:greedy-bound} again, we get
\begin{equation}\label{eq:bound2}
    OPT \leq 
    \frac{3e}{e-1}\left( 1 + \frac{4e}{C (e-1) -6e  +2}\right) V(S_G)
    + \frac{2e\varepsilon}{C(e-1)- 6e + 2}.
\end{equation}
Our choice of $C$, namely,
\begin{equation}\label{eq:constant}
    C =  \frac{8e-1 + \sqrt{64e^2-24e + 9}}{2(e-1)},
\end{equation}
 is precisely to minimize the maximum among the coefficients of $V_{i^*}$  and $V(S_G)$  in \eqref{eq:bound1}
and \eqref{eq:bound2}, respectively. Indeed, consider:
\begin{displaymath}
    \max\left(1+C,\frac{3e}{e-1}\left( 1 + \frac{4e}{C (e-1) -6e  +2}
            \right)\right).
\end{displaymath}
This function has two minima, only one of those is such that $C(e-1) -6e
+2 \geq 0$. This minimum is precisely \eqref{eq:constant}.
For this minimum, $\frac{2e\varepsilon}{C(e-1)- 6e + 2}\leq \varepsilon.$
Placing the expression of $C$ in \eqref{eq:bound1} and \eqref{eq:bound2}
gives the approximation ratio in \eqref{approxbound}, and concludes the proof
of Theorem~\ref{thm:main}.\hspace*{\stretch{1}}\qed

\section{Proof of Theorem \ref{thm:lowerbound}}\label{proofoflowerbound}

Suppose, for contradiction, that such a mechanism exists. From Myerson's Theorem \cite{myerson}, a single parameter auction is truthful if and only if the allocation function is monotone and agents are paid theshold payments. Consider two
experiments with dimension $d=2$, such that $x_1 = e_1=[1 ,0]$, $x_2=e_2=[0,1]$
and $c_1=c_2=B/2+\epsilon$. Then, one of the two experiments, say, $x_1$, must
be in the set selected by the mechanism, otherwise the ratio is unbounded,
a contradiction. If $x_1$ lowers its value to $B/2-\epsilon$, by monotonicity
it remains in the solution; by  threshold payment, it is paid at least
$B/2+\epsilon$. So $x_2$ is not included in the solution by budget feasibility
and individual rationality: hence, the selected set attains a value $\log2$,
while the optimal value is $2\log 2$.\hspace*{\stretch{1}}\qed

\section{Extensions}\label{sec:ext}
\subsection{Strategic Experimental Design with Non-Homotropic Prior}\label{sec:bed}


In the general case where the prior distribution of the experimenter on the
model $\beta$ in \eqref{model} is not homotropic and has a generic covariance
matrix $R$, the value function takes the general form given by
\eqref{dcrit}.

Let us denote by $\lambda$ (resp. $\Lambda$) the smallest (resp. largest)
eigenvalue of $R$, applying the mechanism described in
Algorithm~\ref{mechanism} and adapting the analysis of the approximation ratio
\emph{mutatis mutandis}, we get the following result which extends
Theorem~\ref{thm:main}. 

\begin{theorem}
 For any $\delta\in(0,1]$, and any $\epsilon\in(0,1]$, there exists a $\delta$-truthful, individually rational
 and budget feasible mechanism for the objective function $V$ given by
 \eqref{dcrit} that runs in time
 $O(\text{poly}(n, d, \log\log\frac{B\Lambda}{\varepsilon\delta b\lambda}))$
 and allocates a set $S^*$ such that:
 \begin{displaymath}
    OPT \leq 
    \bigg(\frac{6e-2}{e-1}\frac{1/\lambda}{\log(1+1/\lambda)}+ 4.66\bigg)V(S^*)
    + \varepsilon.
\end{displaymath}
\end{theorem}

\subsection{Non-Bayesian Setting}

In the non-bayesian setting, \emph{i.e.} when the experimenter has no prior
distribution on the model, the covariance matrix $R$ is the zero matrix. In this case,
the ridge regression estimation procedure \eqref{ridge} reduces to simple least squares (\emph{i.e.}, linear regression),
and the $D$-optimality criterion reduces to the entropy of $\hat{\beta}$, given by:
\begin{equation}\label{eq:d-optimal}
V(S) = \log\det(X_S^TX_S)
\end{equation}
A natural question which arises is whether it is possible to design
a deterministic mechanism in this setting. Since \eqref{eq:d-optimal} may take
arbitrarily small negative values, to define a meaningful approximation one
would consider  the (equivalent) maximization of $V(S) = \det\T{X_S}X_S$.
However, the following lower bound implies that such an optimization goal
cannot be attained under the constraints of  truthfulness, budget feasibility,
and individual rationality.

\begin{lemma}
For any $M>1$, there is no $M$-approximate, truthful, budget feasible,
individually rational mechanism for a budget feasible reverse auction with
value function $V(S) = \det{\T{X_S}X_S}$.  
\end{lemma}

\begin{proof}
From Myerson's Theorem \cite{myerson}, a single parameter auction is truthful if and only if the allocation function is monotone and agents are paid theshold payments. Given $M>1$, consider $n=4$ experiments of dimension $d=2$. For $e_1,e_2$ the standard basis vectors in $\reals^2$, let $x_1 = e_1$, $x_2 = e_1$, and $x_3=\delta e_1$, $x_4=\delta e_2$, where $0<\delta<1/(M-1) $. Moreover, assume that $c_1=c_2=0.5+\epsilon$, while $c_3=c_4=\epsilon$, for some small $\epsilon>0$. Suppose, for the sake of contradiction, that there exists a mechanism with approximation ratio $M$. Then, it must include in the solution $S$ at least one of $x_1$ or $x_2$: if not, then $V(S)\leq \delta^2$, while $OPT = (1+\delta)\delta$, a contradiction. Suppose thus that the solution contains $x_1$. By the monotonicity property, if the cost of experiment $x_1$ reduces to $B/2-3\epsilon$, $x_1$ will still be in the solution. By threshold payments, experiment $x_1$ receives in this case a payment that is at least $B/2+\epsilon$. By individual rationality and budget feasibility, $x_2$ cannot be included in the solution, so $V(S)$ is at most $(1+\delta)\delta$. However, the optimal solution includes all experiments, and yields $OPT=(1+\delta)^2$, a contradiction. 

\end{proof}


\subsection{Beyond Linear Models}
Selecting experiments that maximize the information gain in the Bayesian setup
leads to a natural generalization to other learning examples beyond linear
regression. In particular, consider the following variant of the standard PAC learning setup \cite{valiant}: assume that the features $x_i$, $i\in \mathcal{N}$ take values in some generic set $\Omega$, called the \emph{query space}. 
Measurements $y_i\in\reals$ are given by 
 \begin{equation}\label{eq:hypothesis-model}
    y_i = h(x_i) + \varepsilon_i
\end{equation}
where $h\in \mathcal{H}$ for some subset $\mathcal{H}$ of all possible mappings
$h:\Omega\to\reals$, called the \emph{hypothesis space}. As before, we assume that the experimenter has a prior distribution on the
hypothesis $h\in \mathcal{H}$; we also assume that  $\varepsilon_i$
are random variables in $\reals$, not necessarily identically distributed, that
are independent \emph{conditioned on $h$}. As before, the features $x_i$ are public, and the goal of the experimenter is to (a) retrieve measurements $y_i$ and (b) estimate $h$ as accurately as possible.

This model is quite broad, and
captures many classic machine learning tasks; we give a few concrete examples below:
\begin{enumerate}
\item\textbf{Generalized Linear Regression.} In this case, $\Omega=\reals^d$, $\mathcal{H}$ is the set of linear maps $\{h(x) = \T{\beta}x \text{ s.t. } \beta\in \reals^d\}$, and $\varepsilon_i$ are independent zero-mean variables (not necessarily identically distributed). 
\item\textbf{Learning Binary Functions with Bernoulli Noise.} When learning a binary function under noise, the experimenter wishes to determine a binary function $h$ by testing its output on differrent inputs; however, the output may be corrupted with probability $p$. Formally, $\Omega = \{0,1\}^d$, $\mathcal{H}$ is some subset of binary functions $h:\Omega\to\{0,1\}$, and $$\varepsilon_i =\begin{cases}0, &\text{w.~prob.}\;1-p\\\bar{h}(x_i)-h(x_i), &\text{w.~prob.}\;p\end{cases}$$
\item\textbf{Logistic Regression.} Logistic regression aims to learn a hyperplane separating $+1$--labeled values from $-1$--labeled values; again, values can be corrupted, and the probability that a label is flipped drops with the distance from the hyperplane. Formally, $\Omega=\reals^d$,  $\mathcal{H}$ is the set of maps $\{h(x) = \mathop{\mathrm{sign}}(\beta^T x) \text{ for some } \beta\in\reals^d\}$, and $\varepsilon_i$ are independent conditioned on $\beta$ such that $$\varepsilon_i=\begin{cases} -2\cdot\id_{\beta^Tx>0}, & \text{w.~prob.} \frac{1}{1+e^{\beta^Tx}}\;\\ +2\cdot\id_{\beta^Tx<0},&\text{w.~prob.}\frac{e^{\beta^Tx}}{1+e^{\beta^Tx}}\;\end{cases}$$
\end{enumerate}

We can again define the information gain as an objective to maximize:
\begin{align}\label{general}
V(S) = \entropy(h) -\entropy(h\mid y_S),\quad S\subseteq\mathcal{N} 
\end{align}
This is a monotone set function, and it clearly satisfies $V(\emptyset)=0$.
 In general, the information gain is not a submodular function. However, when the errors $\epsilon_i$ are independent conditioned on $h$, the following lemma holds: \begin{lemma}
The value function given by the information gain \eqref{general} is submodular.
\end{lemma}
\begin{proof}
A more general statement for graphical models is shown in~\cite{krause2005near}; in short, using the chain rule for the conditional entropy we get:
\begin{equation}\label{eq:chain-rule}
    V(S) = H(y_S) - H(y_S \mid h) 
    = H(y_S) - \sum_{i\in S} H(y_i \mid h)
\end{equation}
where the second equality comes from the independence of the $y_i$'s
conditioned on $h$. Recall that the joint entropy of a set of random
variables is a submodular function. Thus, the value function is written in
\eqref{eq:chain-rule} as the sum of a submodular function and a modular function.
\end{proof}

This lemma implies that learning an \emph{arbitrary hypothesis, under an
arbitrary prior} when noise is conditionally independent leads to a submodular
value function. Hence, we can apply the previously known results by \citeN{singer-mechanisms} and \citeN{chen} to get
the following corollary:
\begin{corollary}
    For Bayesian experimental design with an objective given by the
    information gain \eqref{general}, there exists a  randomized, polynomial-time, budget
    feasible, individually rational, and universally truthful mechanism with
    a $7.91$ approximation ratio, in expectation.

    In cases where maximizing \eqref{general} can be done in polynomial time in
    the full-information setup, there exists a deterministic,  polynomial-time, budget feasible,
    individually rational, and truthful mechanism for Bayesian experimental
    design with an $8.34$ approximation ratio. 
\end{corollary}

Note however that, in many scenarios covered by this model (including the last
two examples above), even computing the entropy under a given set might be
a hard task---\emph{i.e.}, the value query model may not apply. Hence,
identifying learning tasks in the above class for which truthful or universally
truthful constant approximation mechanisms exist, or studying these problems in
the context of stronger query models such as the demand model
\cite{dobz2011-mechanisms,bei2012budget} remains an interesting open question. 


\section{Non-Truthfulness of the Maximum Operator}\label{sec:non-monotonicity}
We give a counterxample of the truthfulness of the maximum mechanism whose
allocation rule is defined in \eqref{eq:max-algorithm} when the value function
$V$ is as defined in \eqref{obj}. We denote by $(e_1, e_2, e_3)$ the canonical
basis of $\reals^3$ and define the following feature vectors: $x_1=e_1$,
$x_2=\frac{1}{\sqrt{2}}\cos\frac{\pi}{5}e_2
+ \frac{1}{\sqrt{2}}\sin\frac{\pi}{5}e_3$, $x_3=\frac{1}{\sqrt{2}}e_2$ and $x_4
= \frac{1}{2}e_3$, with associated costs $c_1 = \frac{5}{2}$, $c_2=c_3 = 1$ and
$c_4=\frac{2}{3}$. We also assume that the budget of the auctioneer is
$B=\frac{5}{2}$.

Note that $V(x_i) = \log(1+\|x_i\|^2)$, so $x_1$ is the point of maximum value.
Let us now compute the output of the greedy heuristic. We have:
\begin{equation}\label{eq:local-bazinga}
    \frac{V(x_1)}{c_1} \simeq 0.277,\;
    \frac{V(x_2)}{c_2}= \frac{V(x_3)}{c_3} \simeq 0.405,\;
    \frac{V(x_4)}{c_4} \simeq 0.335
\end{equation}
so the greedy heuristic will start by selecting $x_2$ or $x_3$. Without loss of
generality, we can assume that it selected $x_2$. From the Sherman-Morrison
formula we get:
\begin{displaymath}
    V(\{x_i, x_j\}) - V(x_i) = \log\bigg(1+ \|x_j\|^2
    - \frac{\ip{x_i}{x_j}^2}{1+\|x_i\|^2}\bigg)
\end{displaymath}
In particular, when $x_i$ and $x_j$ are orthogonal $V(\{x_i, x_j\}) = V(x_j)$.
This allows us to compute:
\begin{displaymath}
    \frac{V(\{x_2,x_3\})-V(x_2)}{c_3}=\log\bigg(1+\frac{1}{2}
    - \frac{1}{6}\cos^2\frac{\pi}{5}\bigg)\simeq 0.329
\end{displaymath}
\begin{displaymath}
    \frac{V(\{x_2,x_4\})-V(x_2)}{c_4}=\frac{3}{2}\log\bigg(1+\frac{1}{4}
    - \frac{1}{12}\sin^2\frac{\pi}{5}\bigg)\simeq 0.299
\end{displaymath}
Note that at this point $x_1$ cannot be selected without exceding the budget.
Hence, the greedy heuristic will add $x_3$ to the greedy solution and returns
the set $\{x_2, x_3\}$ with value:
\begin{displaymath}
    V(\{x_2, x_3\}) = V(x_2) + V(\{x_2, x_3\}) - V(x_2)\simeq 0.734
\end{displaymath}
In contrast, $V(x_1) \simeq 0.693$ so the mechanism will allocate to $\{x_2,
x_3\}$.

Let us now assume that user $3$ reduces her cost. It follows from
\eqref{eq:local-bazinga} that the greedy heuristic will start by adding her to
the greedy solution. Furthermore:
\begin{displaymath}
    \frac{V(\{x_3,x_2\})-V(x_3)}{c_2}=\log\bigg(1+\frac{1}{2}
    - \frac{1}{6}\cos^2\frac{\pi}{5}\bigg)\simeq 0.329
\end{displaymath}
\begin{displaymath}
    \frac{V(\{x_3,x_4\})-V(x_3)}{c_4}
    =\frac{3}{2}\log\bigg(1+\frac{1}{4}\bigg)\simeq 0.334
\end{displaymath}
Hence, the greedy solution will be $\{x_3, x_4\}$ with value:
\begin{displaymath}
    V(\{x_3, x_4\}) = V(x_3) + V(\{x_3, x_4\}) - V(x_3)\simeq 0.628
\end{displaymath}
As a consequence the mechanism will allocate to user $1$ in this case. By
reducing her cost, user 3, who was previously allocated, is now rejected by the
mechanism. This contradicts the monotonicity of the allocation rule, hence its
truthfulness by Myerson's theorem \cite{myerson}, which states that a single parameter auction is truthful if and only if the allocation function is monotone and agents are paid theshold payments.

\end{document}